\documentclass[12pt]{article}   
\usepackage{arxiv}

\usepackage{orcidlink,thumbpdf,lmodern}

\usepackage{xcolor}
\usepackage{tikz}
\usepackage{float}
\usetikzlibrary{positioning, fit, calc, shapes, arrows, matrix}
\usepackage{subcaption}
\usepackage{amsthm,amsmath,amsfonts,amssymb,amsbsy}
\usepackage{tabularray}
\UseTblrLibrary{booktabs}
\usepackage{etoolbox}
\usepackage{natbib}

\makeatletter
\newcommand\code{\bgroup\@makeother\_\@makeother\~\@makeother\$\@codex}
\def\@codex#1{{\normalfont\ttfamily\hyphenchar\font=-1 #1}\egroup}
\makeatother

\let\proglang=\textsf
\newcommand{\pkg}[1]{{\fontseries{m}\fontseries{b}\selectfont #1}}

\newcommand{\depth}{\texttt{depth\ }}
\newcommand{\lati}{\texttt{lat\ }}
\newcommand{\longi}{\texttt{long\ }}
\newcommand{\stat}{\texttt{stations\ }}
\newcommand{\magn}{\texttt{mag\ }}
\newcommand{\dimX}{d_x}
\newcommand{\dimT}{d_t}
\newcommand{\xX}{\mathbf{x}}
\newcommand{\xt}{t}
\newcommand{\xT}{\mathcal{T}}
\newcommand{\xE}{\mathbb{E}}

\newcommand{\xR}{\mathbb{R}}
\newcommand{\bs}{\boldsymbol}
\newcommand{\proc}[3]{
    \ifstrempty{#3}%
    {%
        \ifstrempty{#2}%
        {%
            #1
        }{%
            #1_{#2}
        }%
    }{%
        (#1_{#2})_{#3}
    }%
}
\newcommand{\slogt}{\Psi} 

\newcommand{\tikzxmark}{%
\tikz[scale=0.23] {
    \draw[line width=0.8,line cap=round] (0,0) to [bend left=14] (1,1);
    \draw[line width=0.7,line cap=round] (0.2,0.95) to [bend right=15] (0.8,0.05);
}}

\newcommand{\tikzcmark}{%
\tikz[scale=0.23] {
    \draw[line width=0.7,line cap=round] (0.25,0) to [bend left=10] (1,1);
    \draw[line width=0.8,line cap=round] (0,0.35) to [bend right=1] (0.23,0);
}}

\theoremstyle{plain}
\newtheorem{theorem}{Theorem}[section]

\newtheorem{proposition}[theorem]{Proposition}

\theoremstyle{definition}

\theoremstyle{remark}
\newtheorem*{remark}{Remark}

\title{ SLGP: An R Package for Spatial Conditional Density Estimation}

\author{Athénaïs Gautier\\
\thanks{Version 2.0.0 of the \pkg{SLGP} package, available from CRAN at \url{https://CRAN.R-project.org/package=SLGP}.}\\   
    DTIS, ONERA, Université Paris-Saclay, 91120, Palaiseau, France\\
    E-mail: \url{athenais.gautier@onera.fr}}

\begin{document}
\maketitle


\begin{abstract}
The \pkg{SLGP} package provides nonparametric estimation of \emph{fields} of probability densities indexed by covariates, from unevenly sized samples with few or no replicates. A Spatial Logistic Gaussian Process applies a logistic density transformation to a finite-rank Gaussian process over a joint index-response domain, yielding a prior over conditional densities that adapts in location, shape, and modality. The contribution of this paper is the
implementation: a finite-rank Random-Fourier-Feature representation that keeps the transformed process numerically tame, a log-concave likelihood that makes the fit a convex problem, and a family of quadrature shortcuts for the normalising integrals that turn a per-observation cost into a per-grid cost.
Three estimation modes are available: MAP, Laplace, and full MCMC. We demonstrate the package throughout on the classical Fiji-Tonga earthquake catalogue, estimating how the distribution of hypocentre depths deforms across the region, and on a discrete-response variant.
\end{abstract}


\section{Introduction}

Estimating a single probability density from data is classical. Estimating a family of densities $p(\cdot \mid \xX)$ that varies smoothly with covariates $\xX$, from samples that are scattered unevenly across the index space and may contain few or no replicates, is not. This is the setting the \pkg{SLGP} package addresses. Figure~\ref{fig:QuakesScatter} further illustrates two features typical of the problems we target. First, the data are spread unevenly across $\xX$: some $\xX$ instances (or bins) hold a handful of observations, others hundreds, and at an arbitrary query point there may be none at all. Second, the shape of the response distribution drifts with $\xX$, not only in terms of mean and variance but also in modalities or skewness, so no single parametric family fits everywhere. Forming a histogram per bin can be an arbitrary task, requires binning the index, discards the smoothness across $\xX$, and breaks down exactly where data are scarce. We want instead a single object that returns a density at any $x$, borrowing strength from neighboring locations. 

\begin{figure}[H]
\centering
    \includegraphics[width=\linewidth]{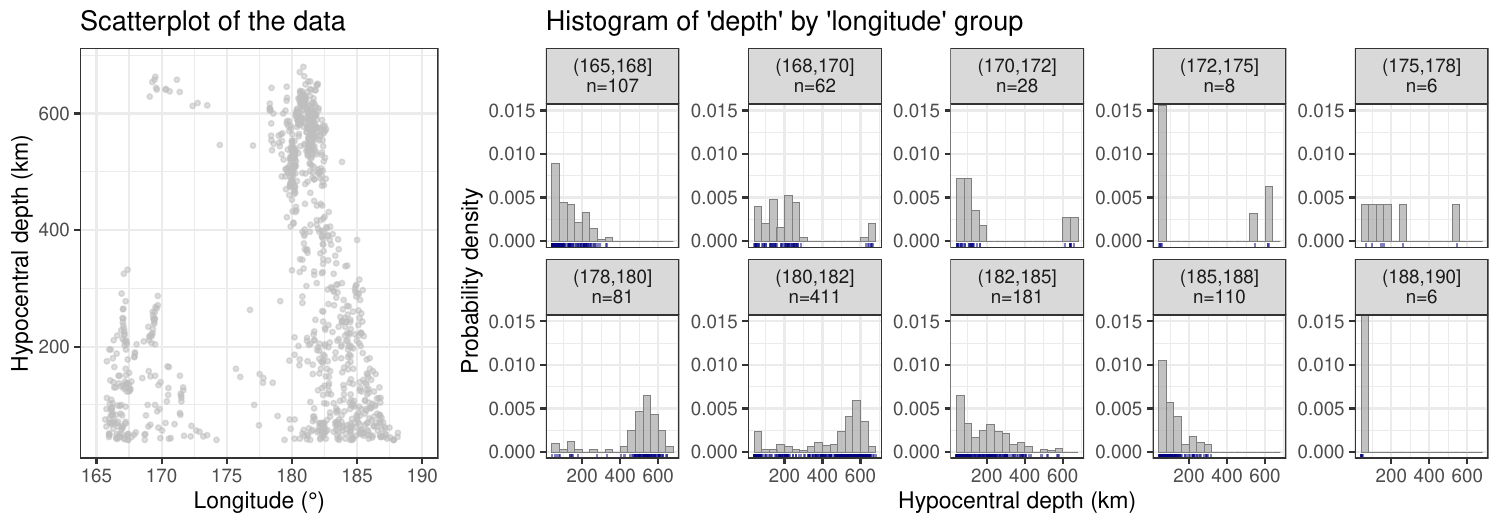}
\caption{A challenging setting in the \code{quakes} dataset: hypocentre \depth against \longi (left), and empirical distributions of \depth within longitude bands (right). Sample sizes per band vary by an order of magnitude, and the distribution of \depth changes shape.}
\label{fig:QuakesScatter}
\end{figure}

Many tools estimate distributions that change with a covariate, but each leaves part of this problem unsolved. Finite mixture models \citep{rojas_conditional_2005, mclachlan_finite_2004} and their Bayesian nonparametric extensions, including Dirichlet process mixtures and related infinite mixture models \citep{jain_split-merge_2004, walker_sampling_2007, papaspiliopoulos_retrospective_2008, ferguson_bayesian_1973, escobar_bayesian_1995}, provide highly flexible conditional densities but typically express the dependence on $\xX$ through mixing weights and often benefit from replicated observations. Conditional kernel density estimators \citep{fan_estimation_1996, hall_methods_1999, silverman_density_1986} are conceptually simple but degrade rapidly when observations are sparse in the covariate space. Similar limitations arise for more general kernel-based conditional distribution estimators \citep{murray_gaussian_2008, rudi_psd_2021, muzellec_learning_2022, donner_efficient_2018}, whose computational cost generally increases with the number of distinct covariate locations. Distributional kriging \citep{aitchison_statistical_1982, egozcue_hilbert_2006, menafoglio_universal_2013, menafoglio_kriging_2014, talska_compositional_2018} interpolates probability distributions over space but typically assumes replicated observations and forms of stationarity that are often unrealistic in heterogeneous settings. Other approaches include generalized lambda distribution regression \citep{zhu_emulation_2020, zhu_surrogate_2023}, which provides a flexible semiparametric family for predominantly unimodal distributions, shape-constrained distributional regression \citep{guntuboyina_nonparametric_2018}, which gains statistical efficiency under strong structural assumptions, and neural conditional density estimators \citep{papamakarios_masked_2017, rothfuss_conditional_2019, papamakarios_neural_2019}, which handle high-dimensional covariates but typically require large training sets and offer limited interpretability or uncertainty quantification. Logistic Gaussian process priors \citep{leonard_density_1978,lenk_logistic_1988,tokdar_posterior_2007,tokdar_bayesian_2010} instead place a genuinely nonparametric prior on a single density by transforming a Gaussian process, while the Spatial Logistic Gaussian Process \citep{gautier_continuous_2021} extends this construction to an entire field of conditional densities indexed by $\xX$.\\
On the software side, several mature \proglang{R} packages cover neighbouring ground. The \pkg{np} package \citep{hayfield_np_2008} implements kernel conditional density and distribution estimation with mixed data types. \pkg{hdrcde} \citep{hyndman_estimating_1996} provides conditional density estimation and highest-density-region graphics. The packages \pkg{gamlss} \citep{rigby_gamlss_2005} and \pkg{bamlss} \citep{umlauf_bamlss_2021} fit distributional regression models in which the parameters of a chosen response family depend flexibly on covariates. Finally \pkg{quantreg} \citep{quantreg} estimates individual conditional quantiles. These tools either impose a parametric response family, target one functional at a time, or require local data abundance. \\
To our knowledge, existing readily available implementations do not simultaneously accommodate scattered and unevenly sized samples, the absence of replicates, and unrestricted variation in distributional shape and modality. That is the gap filled by the SLGP model, and the \pkg{SLGP} package makes it readily accessible in practice.

In line with the objective to demonstrate the potential of SLGP modeling, we adopt as a running example the classical \code{quakes} dataset shipped with base \proglang{R} \citep{r_core_team}. It records $1000$ seismic events of body-wave magnitude above $4.0$ that occurred since 1964 in the Fiji-Tonga region, one of the most seismically active areas in the world; the catalogue originates from the Harvard PRIM-H project. Each event carries its epicentre (\lati, \longi), its hypocentre \depth (between 40 and 680\,km), its magnitude \magn, and the number of \stat that detected it. For the sake of readable slice-based figures, we start with a one-dimensional index and estimate the field of densities of \depth along an east-west transect, indexed by \longi. Section~\ref{subsec:2Dfield} then exploits the full epicentre (\lati, \longi) as a two-dimensional index, a setting where per-bin histograms become untenable but the SLGP applies unchanged. This is an ideal stress test for spatially indexed density estimation. As one moves across the region, the conditional distribution of \depth transitions from unimodal and shallow, to strongly bimodal where shallow and deep seismicity coexist: a regular deformation of shape, spread and modality. Moreover, the sampling is naturally heterogeneous: epicentres cluster along the active structures, so that some cells of the index space contain hundreds of events while others contain none, and little events share the exact same epicentre: there are few replicates. Finally, the ground truth is easy to appraise to the naked eye in a simple scatterplot (Figure~\ref{fig:QuakesScatter}), so readers can validate the fitted field visually against a well-understood geophysical structure.

Before addressing the details and technicalities, Figure~\ref{fig:teaser1} and Figure~\ref{fig:teaser2} show what the package ultimately produces on this dataset. From the $1000$ scattered events, a single call fits the entire field of conditional depth densities across the arc, and a single fitted object then yields densities, quantiles,  moments or the posterior probability of a deep event at any location. The estimate recovers the shallow seismicity near the trench and the deep Wadati-Benioff plane to the west, including in poorly monitored longitudes where no local histogram could be formed \citep{isacks_seismology_1968, frohlich_deep_2006}.

The remainder of the paper explains how this is achieved: the finite-rank representation (Section~\ref{section:practicalElements}) that makes the model tractable, the convex likelihood that makes the fit reliable, and the quadrature that makes it fast, followed by a guided tour of the package (Section~\ref{section:applications}).

\begin{figure}[H]
    \centering
    \begin{subfigure}[t]{\linewidth}
        \centering
        \includegraphics[width=0.9\linewidth]{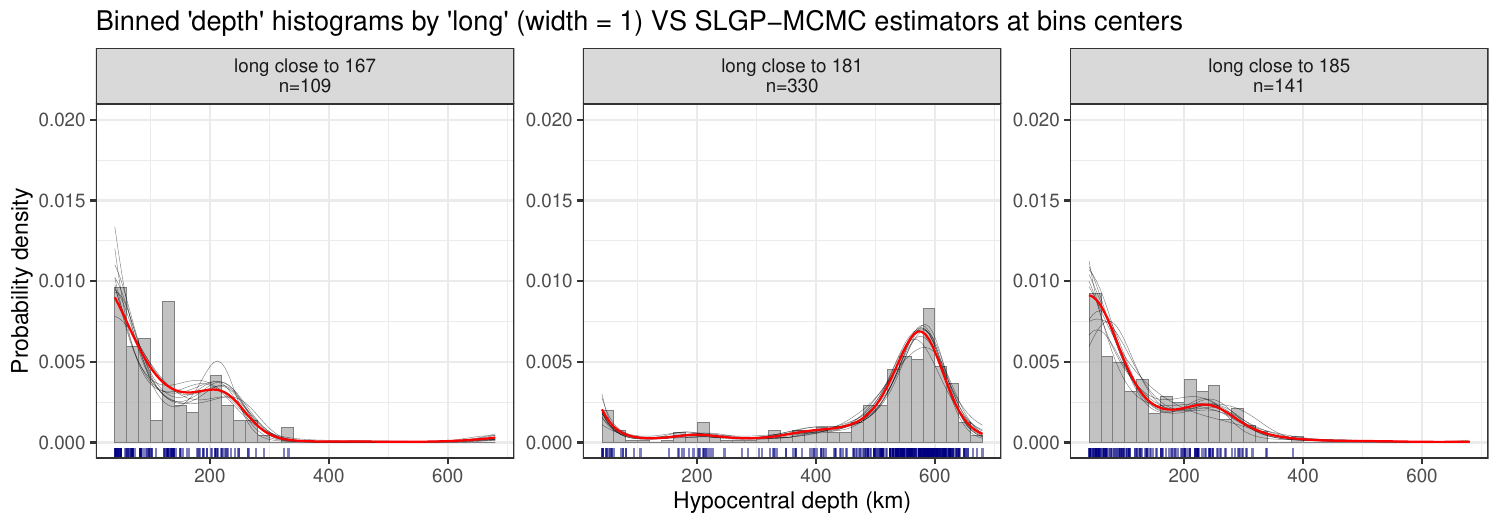}
        \caption{Simultaneous density estimation: posterior predictive depth densities at three longitudes (solid, with credible bands) against binned observations.}
    \label{fig:teaser1a}
  \end{subfigure}

  \begin{subfigure}[t]{\linewidth}
    \centering
    \includegraphics[width=0.9\linewidth]{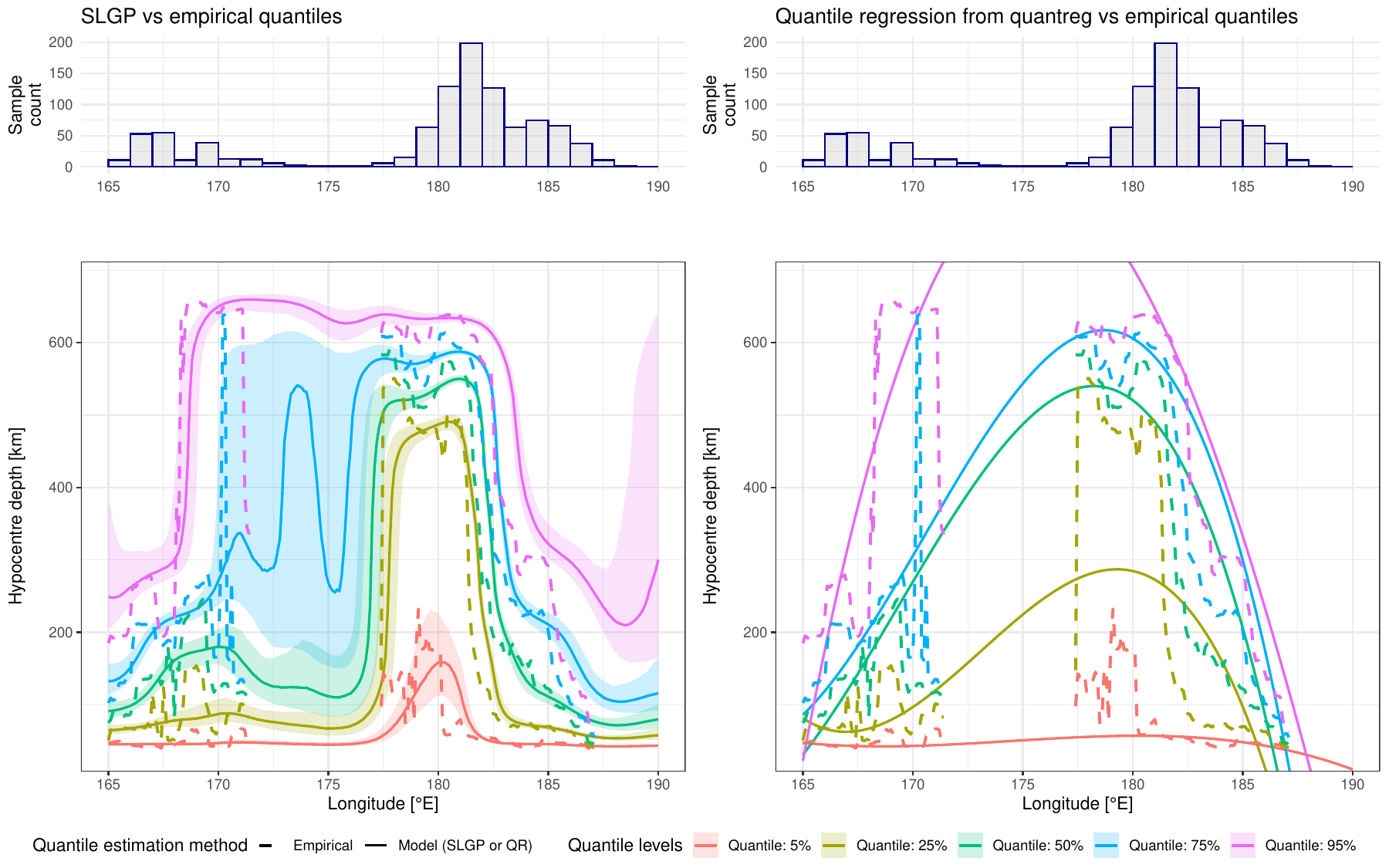}
    \caption{Conditional quantiles of depth across longitude, with posterior uncertainty: median value and 10th-90th percentiles ribbon.}
    \label{fig:teaser1b}
  \end{subfigure}

  \begin{subfigure}[t]{\linewidth}
    \centering
    \includegraphics[width=0.9\linewidth]{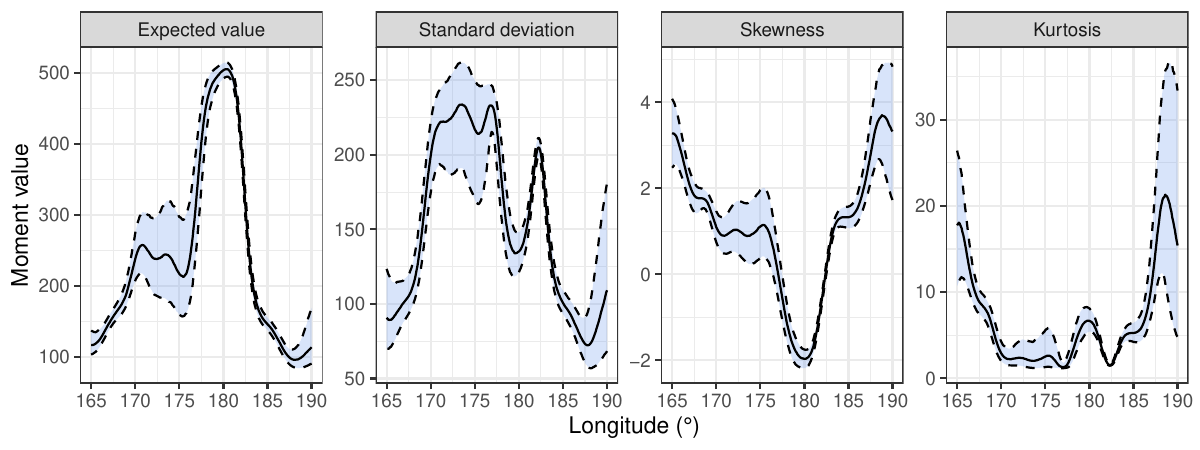}
    \caption{Conditional moments (mean, s.d., skewness, kurtosis) across longitude with uncertainty: median value and 10th-90th percentiles ribbon.}
    \label{fig:teaser1c}
  \end{subfigure}
\caption{One-dimensional index (\code{depth \textasciitilde{} long}): from $1000$ scattered earthquakes, a single SLGP fit yields the full conditional density field (a) and, from the same fitted object, conditional quantiles (b) and moments (c), each with posterior uncertainty. No binning of the data is required.}
\label{fig:teaser1}
\end{figure}

\begin{figure}[H]
\centering
    \includegraphics[width=0.7\linewidth]{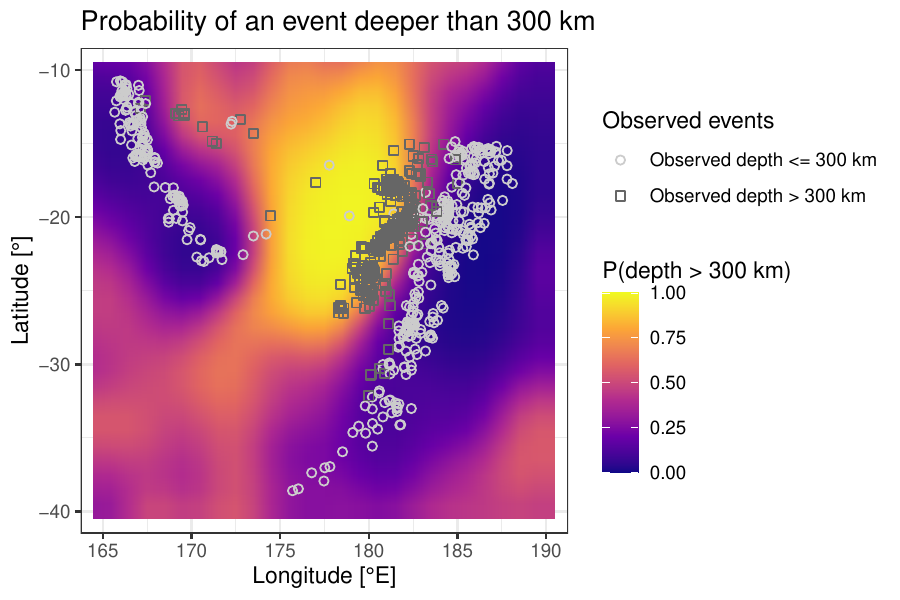}
\caption{Two-dimensional index (\code{depth \textasciitilde{} long + lat}): the posterior probability of a deep event, $P(\text{depth} > 300\,\text{km} \mid \text{long}, \text{lat})$, over the region. This exceedance field derives from the same class of fitted object as the summaries above, now indexed by the full epicentre location.}
\label{fig:teaser2}
\end{figure}
\vspace{-18pt}

\section{The Spatial Logistic Gaussian Process, or SLGP}
\label{section:SLGP}
\subsection{A prior over indexed probability density functions}

The SLGP model, related to \citep{tokdar_bayesian_2010} central to the present contribution, as a continuation of \citep{gautier_goal-oriented_2021, gautier_continuous_2021, gautier_modelling_2023}, is itself a spatial generalization of the Logistic Gaussian Process models, which were established and studied in \citep{lenk_logistic_1988, lenk_towards_1991, leonard_density_1978}.  The SLGP for spatially dependent density estimation builds upon a \textit{well-behaved} GP $\proc{Z}{\xX, t}{(\xX, t) \in D \times \xT}$  and defines the stochastic process obtained from applying the \textit{spatial logistic density} transformation to $Z$ as follows:

\begin{equation}
	\label{eq:informalSLGP}
	\slogt[Z](\xX, t)= \dfrac{e^{Z_{\xX, t}}}{\int_\xT e^{Z_{\xX, u}} \,du } \text{ for all } (\xX, t) \in D \times \xT
\end{equation}
Here and throughout the document, we consider a compact and convex response space $\xT \subset \xR^{\dimT}$ with $\dimT \geq 1$ and we further assume that $\xT$ has a positive Lebesgue measure. We also consider that $D \subset \xR^{\dimX}$ with $\dimX \geq 1$. 

At any fixed $\xX$, $\slogt[Z](\xX, \cdot)$ hence returns a random function that is, by construction, positive and integrates to one. Hence, a SLGP can be (informally) seen as a field of random pdfs. The precise mathematical nature of the objects involved (whether in terms of random measures or densities, and fields thereof) requires careful treatment, for which we refer the reader to \citep{gautier_continuous_2021}.

Here, we adopt a more intuitive approach and assume that the objects are well-behaved enough to speak of SLGP as random fields of densities. We also show in Figure~\ref{fig:PriorSLGP} three draws from a SLGP prior over the running example's domain: fields of densities for \depth, indexed by \longi. Since a SLGP models directly the field as a whole, drawing one realisation of it yields a collection of pdfs indexed by $\xX$. Under mild conditions on the transformed GP, these pdfs will vary smoothly when $\xX$ changes.

\begin{figure}[H]
\centering
    \includegraphics[width=\linewidth]{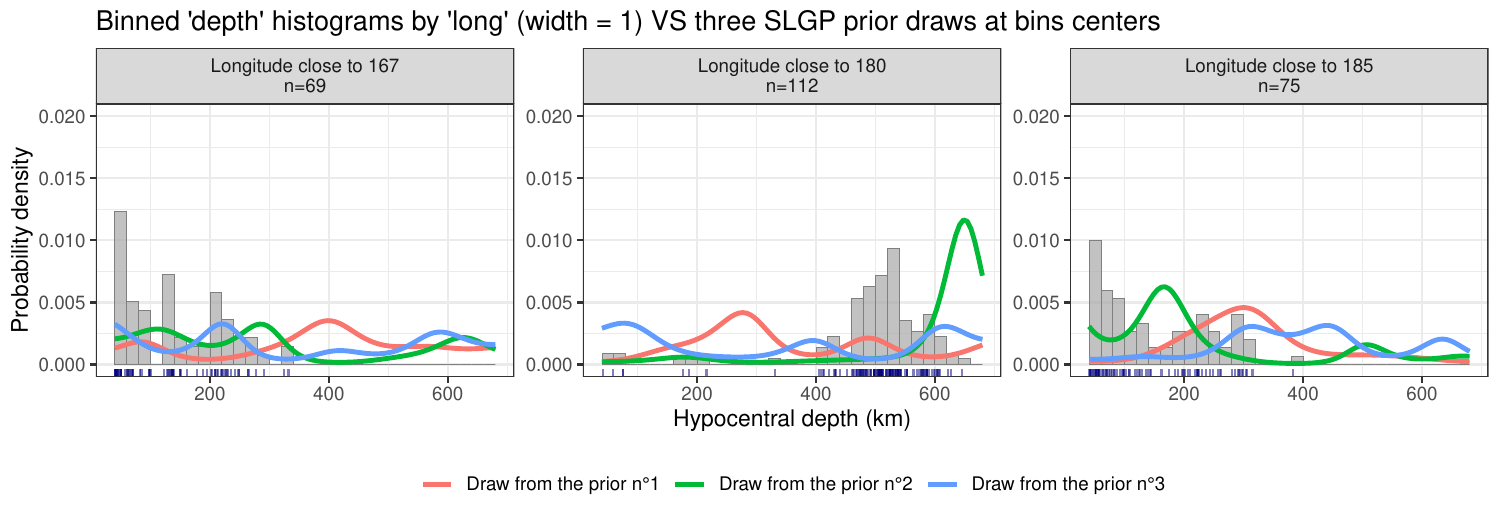}
\caption{Prior densities of \depth given \longi at three representative slices. Binned histograms of nearby observations are shown, since replicates are too scarce for a local histogram.}
\label{fig:PriorSLGP}
\end{figure}

Since these are draws from the prior distribution, they do not yet incorporate information from the actual data. Therefore, it is natural that they do not match the empirical histograms. However, the prior samples display a reasonable level of variability and structure, suggesting that the model would be well-suited for density estimation after incorporating data through posterior inference.

\subsection{Conditioning on data}

We focus on implementation and practicalities. For this purpose, we consider that our observations are obtained by independent sampling of a reference probability density field $p_0$. By that, we mean that our dataset consists of $n$ couples of locations and observations $ \{(\xX_i, \xt_i)\}_{1 \leq i \leq n}$, where the $\xX_i$ are in $[0, 1]^{\dimX}$. Moreover, we assume the $t_i$'s are obtained by independent sampling of random variables $T_i$ with respective densities $p_0(\xX_i, \cdot)$. The random vector of observations is denoted by $\mathbf{T}= (T_i)_{1 \leq i \leq n}$, while we use $\mathbf{\xt} = (\xt_i)_{1 \leq i \leq n}$ for a realisation of $\mathbf{T}$.

Under the assumption that the observations follow the SLGP model and are independent, we obtain a density of observations conditioned on the underlying GP:
\begin{equation}
	\label{eq:int_post}
	\pi \left( \mathbf{T}= \xt \vert Z \right) =  \prod\limits_{i=1}^n \slogt[Z](\xX_i, t_i) =  \prod\limits_{i=1}^n \dfrac{e^{\proc{Z}{\xX_i, t_i}{}}}{\int_{ \xT } e^{\proc{Z}{\xX_i, u}{}} \,du } 
\end{equation}

We insist on the fact that this equation derives naturally from the independence assumption, and that we made no assumptions relative to the index $\xX$. This likelihood is valid for sample locations irregularly spaced, with small sample sizes and possibly few or no replicates.

\begin{remark}
 Note that Equation~\ref{eq:int_post} also applies when the indexing variable $\xX$ is not compactly supported, or even not numeric, as long as one can define a covariance kernel on its domain (and hence a GP and an SLGP).
\end{remark}
 	
This formulation allows us to incorporate observed data into the SLGP framework and carry out Bayesian inference, as it was displayed in Figure~\ref{fig:teaser1}.

\section{Implementation in practice}
\label{section:practicalElements}

Applying Spatial Logistic Gaussian Process (SLGP) models to real-world inference comes with several challenges. Each of the four subsections below states one obstacle and the design choice that resolves it: the infinite-dimensional response integral, addressed by a finite-rank representation and a Fourier-type basis chosen for its behaviour under the exponential (Section 3.1), the choice of an inference paradigm, made reliable by the log-concavity of the likelihood (Section 3.2), the cost of the normalising integrals, reduced by a family of quadrature shortcuts (Section 3.3), and the selection of the GP hyper-parameters (Section 3.4). Together they are what the package contributes over a naive transcription of the model.

\subsection{Addressing the dimensionality of the problem}

The first challenge for our model lies in the fact that the integrals in Equation~\ref{eq:int_post} involve values of $Z$ over the whole response domain. This infinite dimensional object makes likelihood-based computations cumbersome. 

\subsubsection{Finite rank GPs and their use for SLGP modelling}

We propose a way to reduce the dimensionality by considering only finite rank GPs. A GP $Z$ on a generic domain $S$ is called a \emph{finite rank} GP if there exists $p \in \mathbb N$, and a family of functions 	$\left( f_{i} \right)_{1\leq i \leq p}$ on $S$ such that: 
	\begin{equation}
		Z(s) = \sigma \sum_{j=1}^p f_{j}(s) \varepsilon_j   , \ \forall s \in S
	\end{equation}
	where $\bs \varepsilon = \left( \varepsilon_1, ..., \varepsilon_p \right)^\top$ is a random vector of $p$ i.i.d. $\mathcal{N}(0, 1)$ random variables.

\begin{remark}
    Note that from now on, we will use the notation $\varepsilon$ (resp. $\bs \varepsilon$ ) to denote the random variable (resp. random vector), and $\epsilon$ (resp. $\bs \epsilon$) for fixed values and realisations thereof.
\end{remark}

In particular, we will be interested in finite-rank GPs defined over $D \times \xT = [0, 1]^{\dimX+\dimT} $, and as such consider function $\left( f_{i} \right)_{1\leq i \leq p}$ on $[0, 1]^{\dimX+\dimT}$. This yields finite-rank GPs that we can write as follows:
	\begin{equation}
		\label{eq:GPextension}
		Z(\xX, \xt) =  \sigma \sum_{j=1}^p f_{j}(\xX, \xt) \varepsilon_j = \sigma \bs \varepsilon^\top \bs F_{}(\xX, \xt)  , \ \forall \xX \in D, \ \xt\in \xT, \ \
	\end{equation}
	where $\bs F_{}(\xX, \xt):=\left( f_{j}(\xX, \xt) \right)_{1\leq j \leq p}$ is the vector of basis functions evaluated at $(\xX, \xt)$ and $\bs \varepsilon$ is a random vector of $p$ i.i.d. $\mathcal{N}(0, 1)$ random variables.

\begin{remark}
	When the $f_{i}$ are $L^2$ orthonormal, this coincides with the Karhunen-Loève expansion of the GP. However, for most kernels, this expansion is not analytically known, and our application does not benefit from the basis functions being orthonormal.
\end{remark}

\begin{remark}
	A finite rank GP defined as in Equation~\ref{eq:GPextension} has the following covariance kernel:
	\begin{equation}
		\text{Cov}\left( Z(\xX, \xt) , Z(\xX', \xt') \right) = \sigma^2  \sum_{j=1}^p f_{j}(\xX, \xt) f_{j}(\xX', \xt') 
	\end{equation}
\end{remark}

The posterior distribution of $\bs \varepsilon $, given data $\{(\xX_i, t_i)\}_{i=1}^n$ can be derived from Equation~\ref{eq:int_post}. Denoting by $\phi_p$ the pdf of the $p$-variate standard normal distribution, we have:
\begin{equation}
		\label{eq:GPposterior}
	\pi[\bs \epsilon \vert  \mathbf{T} = \mathbf{t} ] \propto   \phi_p\left( \bs\epsilon \right)  \prod\limits_{i=1}^n \dfrac{e^{\sigma \sum_{j=1}^p \epsilon_j f_{j}(\xX_i, t_i)}}{ \int_\xT e^{\sigma \sum_{j=1}^p \epsilon_j f_{j}(\xX_i, u)}  \,du}
\end{equation}

\subsubsection{Basis functions considered}

Because the SLGP exponentiates the latent process, the choice of basis is a
question of numerical stability before it is a question of approximation power.
A basis that is suitable in a linear model can be unusable inside $\exp(\cdot)$: Figure~\ref{fig:badBasisFun} shows finite-rank draws built from Legendre polynomials swinging between near-flat and spuriously spiked densities and leaving floating-point range as the rank grows. This rules out the obvious orthonormal-polynomial choice and motivates the bounded, oscillatory Fourier-type bases that we adopt, and that repeated numerical experiments identified as particularly well behaved under the transformation. 

\begin{figure}[!h]
    \centering
    \includegraphics[width=0.95\linewidth]{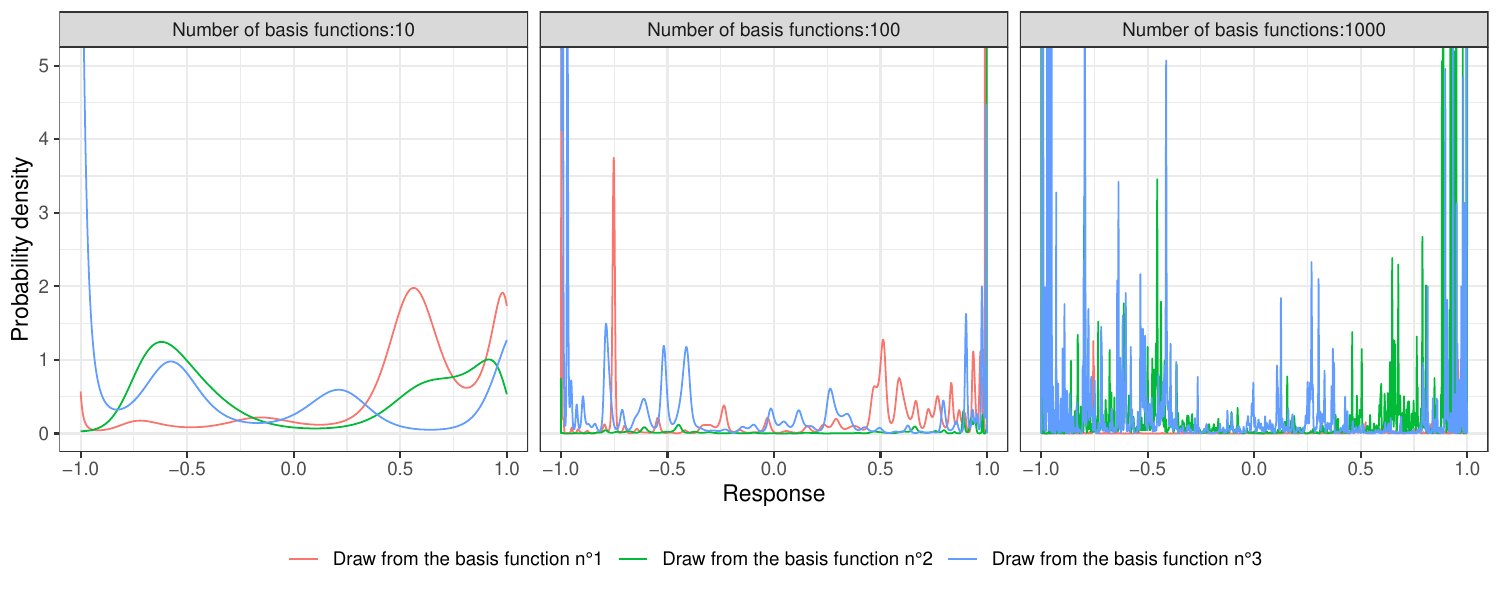}
\caption{Three draws of a finite-rank SLGP indexed at the same value of $\xX$, with basis functions being based on the Legendre Polynomials of order up to $p-1$.}
\label{fig:badBasisFun}
\end{figure}

\vspace{-6pt}

From here on, and to simplify notations, we focus on the setting where
$\xT = [0,1]$, i.e. $\dimT = 1$, the general case follows by replacing
$\dimX+1$ with $\dimX+\dimT$ throughout.

\vspace{-6pt}
\paragraph{Leveraging inducing points}

In \cite{tokdar_towards_2007}, the authors focus on density estimation. To this extent, they rely on LGPs (the non-spatialised counterpart of SLGPs). In their work, they propose to address the dimensionality by considering a reasonable number of inducing points. Instead of working with a GP $Z=\proc{Z(t)}{}{t\in \xT} \sim \mathcal{GP}(0, k)$, their models approximate it by $W(t) := \xE[Z(t) \vert Z(t_1), ..., Z(t_p)]$ (for $p \geq 1$, $t, t_1, ..., t_p \in \xT$).
\begin{proposition}
    The GP $W(t) := \xE[Z(t) \vert Z(t_1), ..., Z(t_p)]$ (for $p \geq 1$, $t, t_1, ..., t_p \in \xT$) is a finite rank GP over $\xT$.
\end{proposition}
\begin{proof}
Let us consider a Gaussian random vector of size $p$, $\mathbf{Z}_p:=(Z(t_1), ..., Z(t_p))^\top$. Let us also denote $K = (k(t_i, t_j))_{1 \leq i, j \leq p}$ the covariance matrix of the chosen design and $k_p(t) = (k(t_i, t))_{1 \leq i \leq p}  (t \in \xT) $. Then, assuming $K$ is positive definite, for any $t\in \xT$:

\begin{align}
    W(t) &= k_p(t)^\top K^{-1} \mathbf{Z}_p = k_p(t)^\top  K^{-1/2} \mathbf{X}_p
\end{align}
where $\mathbf{X}_p = K^{-1/2} \mathbf{Z}_p $ is multivariate standard normal. 

Therefore, setting $f_i(t)$ to be the $i$-th coordinate of the vector $k_p(t)^\top K^{-1/2} $ yields that $W(t) = \sum_{i=1}^p X_i f_i(t)$ where the $X_i$'s are i.i.d. $\mathcal{N}(0, 1)$. 
\end{proof}

We implemented this approach to defining finite-rank SLGPs but we believe that conditioning GPs does not scale well, as the number of inducing points required would blow-up with dimensionality increase. Since the approach requires to compute the inverse square root of a covariance matrix, the numerical stability of the approach would suffer from an increase of the size of this matrix.\\

Instead, we leverage Fourier-type bases, i.e. collections of: $(\cos(\bs \omega^\top [\xX, t]), \sin(\bs \omega^\top [\xX, t]))$, for varying $\bs \omega$'s in $ \xR^{\dimX+1}$. We implement different ways to select the angular frequencies $\bs \omega$. The first one being inspired from the discrete Fourier basis, while the next two are variations within the Random Fourier Feature framework. We also let users provide their custom set of frequencies to be used.

\paragraph{Other frequency sets}
Besides the Random Fourier Features described below, which are the default and the only ones used in this paper, the package implements a multivariate extension of the discrete Fourier basis, in which the frequencies form a regular grid, and a space-filling variant of RFF, in which they are chosen by a
space-filling design with respect to the spectral density rather than sampled
independently. Users may also supply their own set of frequencies. We refer to
the package documentation for these variants, and to Section~\ref{section:benchmark} for the effect of the number of frequencies.

\paragraph{Random Fourier Features}
The framework of Random Fourier Features \citep{rahimi_random_2008,rahimi_weighted_2009} yields another way to construct finite rank GPs that ``resemble'' GPs with a prescribed (stationary and continuous) covariance kernel $k$. Indeed, any such kernel satisfies the hypothesis for the Bochner theorem, a result that relates positive definite, continuous, stationary kernels to the Fourier transform of a finite, positive measure. Whenever the considered measure admits a density $s(\boldsymbol \xi)$ with respect to Lebesgue measure, this gives rise to the so-called Fourier duality of spectral densities and covariance functions. This result is also known as the Wiener-Khinchin theorem \citep{khinchin_korrelationstheorie_1934}: 
\begin{equation}
	\label{eq:WKtheor}
	k_0([\xX, t] -[\xX', t']) = \int_{\xR^{\dimX+1}} s(\boldsymbol \xi) e^{i 2 \pi \boldsymbol \xi^\top [\xX - \xX', t-t']} \,d\boldsymbol \xi
\end{equation}
\begin{equation}
	s(\boldsymbol \xi) = 	\int_{\xR^{\dimX+1}} k_0(\boldsymbol r) e^{ - i 2 \pi \boldsymbol \xi^\top \boldsymbol r} \,d \boldsymbol r
\end{equation}

The Wiener-Khinchin integral (Equation~\ref{eq:WKtheor}) can be re-written as:

\begin{align}
    k_0([\xX, t] -[\xX', t']) =
 	\int_{\xR^{\dimX+1}} s(\boldsymbol \xi) &\left[ \cos  2 \pi 
 \boldsymbol \xi^\top [\xX, t] \cos  2 \pi 
 \boldsymbol \xi^\top [\xX', t']  \right. + \left.\sin  2 \pi 
 \boldsymbol \xi^\top [\xX, t] \sin  2 \pi 
 \boldsymbol \xi^\top [\xX', t'] \right]  \,d\boldsymbol \xi\label{eq:FF1}
\end{align}
The essential element of the approach of Random Fourier Features
\citep{rahimi_random_2008, rahimi_weighted_2009} is the realisation that the
Wiener-Khinchin integral~(\ref{eq:WKtheor}) can be approximated by Monte
Carlo sums. We assume throughout that $k_0$ is normalised so that
$k_0(\bs 0) = 1$, in which case $s$ is a probability density on
$\xR^{\dimX+\dimT}$ and can be sampled from.

Consider $q$ i.i.d. draws $\bs \xi_1, \dots, \bs \xi_q$ from $s$, and define
the rank-$2q$ feature map
\begin{equation}
\label{eq:familyFF}
\bs \varphi([\xX, t]) = \frac{1}{\sqrt{q}}
\Big[
\cos(2\pi \bs \xi_1^\top[\xX, t]), \dots, \cos(2\pi \bs \xi_q^\top[\xX, t]),
\sin(2\pi \bs \xi_1^\top[\xX, t]), \dots, \sin(2\pi \bs \xi_q^\top[\xX, t])
\Big]^\top .
\end{equation}
Then, by Equation~\ref{eq:FF1},
$\mathbb{E}_{\bs \xi_1, \dots, \bs \xi_q}
 \big[ \bs \varphi([\xX, t])^\top \bs \varphi([\xX', t']) \big]
 = k_0([\xX, t] - [\xX', t'])$, which gives
\begin{equation}
\label{eq:kernelRFF}
k([\xX, t], [\xX', t']) = \sigma^2 k_0([\xX, t] - [\xX', t'])
\approx k_{RFF}([\xX, t], [\xX', t'])
:= \sigma^2 \bs \varphi([\xX, t])^\top \bs \varphi([\xX', t']) .
\end{equation}
Moreover, for $\bs \varepsilon$ a $2q$-variate standard normal vector, the
process $Z_{RFF, \xX, t} = \sigma \bs \varphi([\xX, t])^\top \bs \varepsilon$
is a finite-rank GP of rank $p = 2q$, with mean zero and covariance kernel
$k_{RFF}$. In the package, $q$ is the argument \code{nFreq} and $p = 2q$ is
the rank reported by \code{print()}.

\begin{remark}
    In the literature it is common to encounter another basis expansion. Indeed, it is also possible to samples $p$ random phases $U_i \sim \mathcal{U}([0, 2\pi])$ and to use a basis using only cosines:
    
\begin{equation}
    \tilde \varphi(\mathbf y)=\sqrt{\frac{2}{p}} \left[  \cos( 2\pi\boldsymbol \xi_1^\top \mathbf y + U_1), \cos( 2\pi\boldsymbol \xi_2^\top \mathbf y + U_2),...,\cos(2\pi \boldsymbol \xi_p^\top \mathbf y + U_p)\right],
\end{equation}

Philosophically, there is no fundamental difference between these two basis, since $\sin(t+\pi/2)=\cos(t)$. Here, we select the parametrisation in sines/cosines due to its improved behaviour compared to collections of the type $(\cos(2\pi\boldsymbol \xi^\top \mathbf y+ U)$ \citep{sutherland_error_2015}.
\end{remark}

\subsection{Estimation schemes considered and implemented}

Once a family of basis functions is selected, we need to select a suitable estimation paradigm. We propose and implement three estimation strategies. These methods' sanity are underlined by the following result:

\begin{theorem}
	\label{th_log-concavity}
    For fixed data $\mathbf{t}$ the likelihood function $\bs \epsilon \mapsto \mathcal{L}( \bs \epsilon \vert \mathbf{t} ) =  \prod\limits_{j=1}^n \dfrac{e^{\sigma \sum_{i=1}^p \epsilon_i f_{i}(\xX_j, t_j)}}{ \int_\xT e^{\sigma  \sum_{i=1}^p \epsilon_i f_{i}(\xX_j, u)}  \,du} $ is log-concave.	Equivalently, the negative log-likelihood function $\bs \epsilon \mapsto \ell( \bs \epsilon \vert \mathbf{t} )$
	is convex.
\end{theorem}

\begin{proof}[Sketch of proof]
	Note that $\ell( \bs \epsilon \vert \mathbf{t} )$ is twice differentiable with respect to $\bs \epsilon$. Its gradient and Hessian are derived in Appendix~\ref{Appendix:convexity}, and the latter is positive-semidefinite.
\end{proof}

In light of this result, we propose three approaches, summarised in
Table~\ref{table:implementations}.

\paragraph{Markov Chain Monte Carlo (MCMC)}

This method explores the posterior by drawing from it, and enables full posterior inference of $\bs \varepsilon$ given the data. Its drawbacks are the usual ones: hyper-parameters of the sampler to tune, outputs to monitor, and a high computational cost, since each step evaluates the posterior of Equation~\ref{eq:GPposterior}. Our setting benefits from the posterior being a log-concave $\bs \varepsilon$, which ensures rapid convergence of the algorithms considered \citep{dwivedi_log-concave_2018}, and our implementation relies on the No-U-Turn Sampler \citep{hoffman_NUTS_2014} of \pkg{rstan}
\citep{stan_development_team_rstan_2024}, which requires no manual adjustment of the sampling parameters. The cost of gradient evaluations remains, which leads us to also propose faster techniques.

\paragraph{Maximum A Posteriori (MAP)}

MAP estimation returns the mode $\bs \epsilon^*$ of the posterior. It is the fastest scheme we propose, but it yields a non-probabilistic estimate and hence no uncertainty quantification. The optimisation is also performed through \pkg{rstan}, by a quasi-Newton method that exploits the convexity established above.
    
\paragraph{MAP coupled with Laplace approximation} 
The Laplace approximation returns $N$ draws from the multivariate normal distribution centred at the mode of the posterior in Equation~\ref{eq:GPposterior}, with precision matrix the Hessian of the negative log-posterior at the mode, which the package evaluates in closed form (Appendix~\ref{Appendix:convexity}). It is a middle ground between MCMC and MAP: it captures the curvature of the posterior at a small multiple of the cost of a point estimate.

Each scheme returns one or several $\bs \epsilon$, and the corresponding
predictions of the density field are obtained by the same expression,
\begin{equation}
\label{eq:predictors}
\proc{\hat Y}{\xX, t}{}^{(i)} = \dfrac{e^{\sigma \sum_{j=1}^p f_j(x_{1}, ..., x_{\dimX}, t) \epsilon_{i, j} }}{\int_\xT e^{\sigma \sum_{j=1}^p f_j(x_{1}, ..., x_{\dimX}, u) \epsilon_{i, j}} \,du },
\end{equation}
with a single $\bs \epsilon^*$ for MAP, and $N$ draws for Laplace and MCMC.

In most use cases Laplace is a reasonable choice: it quantifies uncertainty and, as measured in Section~\ref{subsec:cost}, costs two to three times a MAP fit against several hundred times for MCMC.

\begin{table}[H]
\centering
\begin{tblr}{
  colspec = {Q[l,2.75cm] Q[c,3.5cm] Q[c,3.5cm] Q[c,4.0cm]}, 
  row{1} = {font=\bfseries}, 
  column{1} = {font=\bfseries}, 
  cells = {valign=m, halign=c}, 
  vline{2-6} = {1}{},
  vline{-} = {2-6}{},
  hline{1} = {2-6}{},
  hline{2-6} = {-}{},
}
                          & MCMC & MAP & MAP with Laplace \\
Probabilistic predictions & {\tikzcmark\\Full post. sampling\\} & {\tikzxmark\\None\\} &  {\tikzcmark\\Normal approximation}   \\
Speed          & {Slowest\\$\approx 500-3000\times$ MAP} &   {Fastest\\reference}  & {Fast\\$\approx 2-3\times$ MAP} \\
Suitable for large $p$    &{\tikzcmark\\}&{\tikzcmark\\}&{\tikzxmark\\Needs to invert Hessian} \\
Typical use case    & {For uncertainty quantification and exact inference.}& {Fast point estimate.\\Start point for MCMC.}&{Compromise between the other two methods.} \\
\end{tblr}
\caption{\label{table:implementations}
Summarising the specifics and suggested use cases of each estimation method.}
\end{table}

\vspace{-18pt}

\subsection{Approximation paradigms for the integrals in the SLGP}
Every estimation mode evaluates the normalising integral many times, and a naive implementation of the likelihood computes one such integral per observation. We can, however, leverage the structure introduced to make this far cheaper.  By utilizing the finite rank representation, we are able to pre-calculate the basis functions at designated nodes. This preparation streamlines the computational process, allowing for rapid substitution of $\bs \varepsilon$ values and subsequent calculation of the requisite integrals. Moreover, the integrand depends on the covariate alone, so the work scales with the number of \emph{distinct} covariate values, not with the sample size, and because the integrand is smooth, its value at one location can be interpolated from a fixed grid. The schemes below (exact, nearest-neighbour and weighted-nearest-neighbour) trade accuracy for speed.
For the sake of brevity and given that the multivariate normal distribution is well known, we choose to center the remaining of our discussion on the intricacies of the likelihood function. The likelihood, which underpins the posterior distribution, is:

\begin{align}
    \mathcal{L}(\bs{\epsilon} \vert  D) &= \prod\limits_{i=1}^n \dfrac{\exp \{\sigma \boldsymbol \epsilon^\top  \boldsymbol F(\xX_i, t_i)\}}{ \int_0^1 \exp \{\sigma \boldsymbol \epsilon^\top  \boldsymbol F(\xX_i, u)\}  \,du} 
\end{align}

Assuming that $(u_i)_{i=1}^m$ are nodes of a regular grid of $[0, 1]$, our proposed quadrature is:

\begin{align}
\label{eq:quadratureInt}
    \mathcal{L}(\bs{\epsilon} \vert  D) & \approx \prod\limits_{i=1}^n \dfrac{\exp \{\sigma \boldsymbol \epsilon^\top  \boldsymbol F(\xX_i, t_i)\}}{ \frac{1}{m} \sum^{m}_{j=1} \exp \{\sigma \boldsymbol \epsilon^\top  \boldsymbol F(\xX_i, u_j)\} }
\end{align}

A first step in ensuring numerical stability of our implementation is to work with the negative log-likelihood rather than with the likelihood itself, as is standard in Bayesian inference. Following our approach and notations from the proof of Theorem~\ref{th_log-concavity}, we need to evaluate:
  
	\begin{align}
		\ell( \bs \epsilon \vert \mathbf{t} ) \approx
		 - \sigma \bs{\epsilon}^{\top} \left(\sum_{i=1}^n \bs F(\xX_i, t_i) \right)+ \sum_{k=1}^K n_k \log \left( \displaystyle \frac{1}{m} \sum^{m}_{j=1} \exp \{\sigma \boldsymbol \epsilon^\top  \boldsymbol F(\tilde \xX_k, u_j)\} \right)
	\end{align}

Evaluating each of the $K$ terms in the second sum can be expensive on large
datasets. We therefore propose three strategies, summarised in
Table~\ref{table:implementationIntegrals}: an \emph{exact} evaluation of the sum at each $\tilde \xX_k$, a \emph{nearest-neighbour} approximation, which
discretises the space of $\xX$ and reuses the integral of the closest
precomputed node, and an \emph{interpolated} approximation, recommended by
default, which interpolates the values at the vertices of the hypercube
surrounding $\tilde \xX_k$.

These three approaches allow users to adapt computations to their needs, depending on whether precision or speed is the priority. They are summarised in Table~\ref{table:implementationIntegrals}, and Figure~\ref{fig:main} of Appendix~\ref{Appendix:approxIntegral} illustrates where each of them evaluates the basis functions. For the equations behind these approximation schemes, we refer the reader to that same appendix.

\begin{table}[H]
\centering
\begin{tblr}{
  colspec = {Q[l,2.75cm] Q[c,3.5cm] Q[c,3.5cm] Q[c,4cm]}, 
  row{1} = {font=\bfseries}, 
  column{1} = {font=\bfseries}, 
  cells = {valign=m, halign=c}, 
  vline{2-6} = {1}{},
  vline{-} = {2-6}{},
  hline{1} = {2-6}{},
  hline{2-6} = {-}{},
}
                          & No approximation & NN approximation & WNN approximation \\
Estimation quality & {\tikzcmark\\Exact} & {\tikzxmark\\Piecewise constant} &  {$\sim$\\Interpolation}   \\

Estimation speed  & {Slowest} &   Fastest  & Fast \\

Typical use case  & { With few different indices $\xX$ } & {To get the fastest estimate} & {Compromise between the other two methods.} 
\end{tblr}
\caption{\label{table:implementationIntegrals}
Summarising the specifics of each integral approximation scheme.}
\end{table}

\vspace{-18pt}

\subsection{Tackling the hyper-parameters in SLGP estimation.}
\label{subsec:hyperparam}
Yet another problem posed is that in almost all realistic cases, GPs depend on some unknown hyperparameters that need to be estimated. This issue could be addressed in a Bayesian way, by specifying a prior on the hyperparameters. While a principled approach to selecting suitable priors have been proposed and studied, \citep{berger_objective_2001, fuglstad_constructing_2019}, these generally do not allow for Theorem~\ref{th_log-concavity} to hold. As this log-concavity is crucial to the efficiency of our estimation procedure, we rather focus on heuristics and simple estimation procedures.

Assuming that the GP is parametrised by a variance parameter $\sigma^2$ as well as as well as $\dimX+1$ length-scale parameters $\rho_1, ..., \rho_{\dimX}, \rho_t$ through: 
\begin{equation}
	\label{eq:GPextensionPara}
	\proc{Z}{\xX, t}{} = \sigma \sum_{j=1}^p f_j(x_1 / \rho_1, ..., x_{\dimX} / \rho_{\dimX}, t/\rho_t) \varepsilon_j , \ \forall \xX \in D, \ t\in \xT
\end{equation}
we propose using a heuristic for $\sigma^2$ and either another heuristic or a grid search for the $\rho$'s.

\textbf{Variance heuristic: } This part consists in numerically selecting a variance that ensures numerical stability of the SLGP. In our case, we typically select $\sigma^2$ such that $\mathbb{E}[\max_{\xX \in D} ( \max_{t \in \xT} \proc{Z}{\xX, t}{} - \min_{t \in \xT} \proc{Z}{\xX, t}{}) ] \leq 5$. This heuristic controls the dynamic range of the field: at any given index, the ratio between the largest and the smallest value of the estimated density is then typically at most $e^5 \approx 148$. This step is not performed conditionally on any data, and is here to ensure that we specified a model that does not present numeric overflow.

Motivating examples underlying this heuristic are available in Appendix, Figure~\ref{fig:Heuristic_var}, where the impact of $\sigma$ on the prior's expressiveness and stability can be visualised. Smaller values of $\sigma$ yield distribution fields that resemble the uniform, while the larger values produce sharply peaked fields that might produce numerical instabilities. It appears that ensuring $\max_{\xX \in D} ( \max_{t \in \xT} \proc{Z}{\xX, t}{} - \min_{t \in \xT} \proc{Z}{\xX, t}{}) \approx 5$ yields a compromise between expressiveness and stability.

\textbf{Length-scale selection: } As for the length-scales, we propose two approaches: a heuristic one, motivated by extensive model fitting over varied datasets, and a MAP-based approach.

The former is motivated by the fact that in practice, and especially in small data regime, a useful rule of thumb is to set each length-scale to about 10–15\% of the corresponding domain size.  Otherwise, a grid search over candidate values, followed by selection based on the profile posterior is advised. To this end, we specify a prior $\pi$ over $\boldsymbol \rho$, and focus on the joint posterior:

\begin{equation}
	\label{eq:int_post_joint_finite}
\pi( \boldsymbol \epsilon; \boldsymbol \rho \vert \mathbf{T}_n = \mathbf{t}_n) \propto \pi(\boldsymbol \rho) \phi( \boldsymbol \epsilon) \prod\limits_{i=1}^n
	\dfrac{e^{\sigma \sum_{j=1}^p f_j(x_{1, i} / \rho_{1}, ..., x_{\dimX, i} / \rho_{\dimX}, t_i/\rho_t) \epsilon_j }}{\int_\xT e^{\sigma \sum_{j=1}^p f_j(x_{1, i} / \rho_1, ..., x_{\dimX, i} / \rho_{\dimX}, u/\rho_t) \epsilon_j} \,du } 
\end{equation}
where, in light of the results from \cite{van_der_vaart_adaptive_2009}, we suggest an inverse Gamma prior for $\boldsymbol \rho$.

Optimising the joint prior for a fixed $\bs \rho$ is a convex problem that yields a unique optimal predictor $\bs \epsilon^*(\bs  \rho) = \arg\max_{\bs \epsilon \in \xR^p}  \pi( \boldsymbol \epsilon; \boldsymbol \rho \vert \mathbf{T}_n = \mathbf{t}_n) $. We then identify the optimal length-scales as the maximisers of the profile posterior: $\bs \rho^* \in \arg\max_{\bs \rho \in (0, \infty)^{\dimX+1}}  \pi( \bs \epsilon^*(\bs  \rho); \boldsymbol \rho \vert \mathbf{T}_n = \mathbf{t}_n) $. In practice, this outer optimisation over $\boldsymbol \rho$ is carried out via a simple grid search. We further illustrate the impact of length-scale selection in the Appendix~\ref{app:lengthscale_optim}.

\section{Guidelines for the users, application to the Fiji-Tonga earthquakes}
\label{section:applications}

All subsequent codes are executed relying on the \proglang{R} \pkg{SLGP} package \citep{SLGP_R}, and a detailed markdown document to reproduce all the Figures and numerical development in this document is available on the author's Github page \citep{gautier_athenais_github_imple}. The dataset is loaded with \code{data("quakes")}, and some visuals are produced with the \pkg{ggplot2} package \citep{ggplot2}. SLGP models are manipulated through the standard generics provided by the package (\code{print()}, \code{summary()}, \code{plot()}, \code{predict()}, \code{simulate()}, \code{coef()}, \code{nobs()}, \code{update()} and \code{timing()}, the last reporting the cost of a fit). \code{summary()} also accepts \code{diagnostics = TRUE}, which evaluates the fitted field on the data or on held-out observations, and we distinguish the default visuals from those produced with \pkg{ggplot2}.

\subsection{Density field estimation with a SLGP}
\subsubsection{Prior induced by a SLGP model}
To build intuition about the behavior of a SLGP prior, we begin by visualising some realisations. Specifically, we approximate a Matérn 5/2 kernel using 200 sine and 200 cosine basis functions in a RFF framework. The length-scales are set to 15\% of the range of the inputs, ensuring a reasonable degree of smoothness while capturing variability across the domain. The variance is automatically determined based on the heuristic we proposed earlier. Below is an \proglang{R} code snippet illustrating how to generate such a prior, followed by the resulting illustration in Figure~\ref{fig:PriorSLGP}:

\begin{verbatim}
R> library(SLGP); data("quakes"); df <- quakes 
R> range_response <- c(40, 680); range_x <- c(165, 190)
R> modelPrior <- slgp(depth~long, # Formula to specify predictors VS response
+                     data=df, method="Prior",
+                     basisFunctionsUsed = "RFF", interpolateBasisFun="WNN", 
+                     hyperparams = list(lengthscale=c(0.15, 0.15), sigma2=1), 
+                     sigmaEstimationMethod = "heuristic", # rewrites sigma2
+                     predictorsLower= range_x[1], predictorsUpper= range_x[2],
+                     responseRange= c(40, 680), seed=1, opts=list(ndraws = 3),
+                     opts_BasisFun = list(nFreq=200, MatParam=5/2))
\end{verbatim}

The fitted object has a \texttt{print()} method giving a short description of the model: its formula, estimation method, basis and rank, response range, data dimensions and hyperparameters. 

\begin{verbatim}
R> print(modelPrior)
Spatial Logistic Gaussian Process (SLGP) model
  Formula        : depth ~ long
  Response       : depth  in [40, 680]
  Covariate(s)   : long
  Estimation     : Prior
  Basis          : RFF  (rank p = 400)
  Length-scales  : 0.15, 0.15
  Variance sigma2: 1.124
  Data           : 1000 obs, 605 distinct covariate value(s)
  Fitting time   : 1.08s
  Coefficient draws: 3
\end{verbatim}

To understand how the SLGP prior behaves, we visualise the three draws it carries. The \texttt{plot()} method for \texttt{SLGP} objects gives a quick look, in the style of Figure~\ref{fig:LAPLACE_draws}. The \texttt{predict()} method returns the density field on arbitrary data, from which we produce a richer figure.

\begin{verbatim}
R> dfGrid <- data.frame(expand.grid(seq(range_x[1], range_x[2], 1), 
+                           seq(range_response[1], range_response[2],, 101)))
R> colnames(dfGrid) <- c("long", "depth")
R> predPrior <- predict(modelPrior, newdata = dfGrid, type="density")
\end{verbatim} 

The result is Figure~\ref{fig:PriorSLGPggplot}. The full plotting code is available in the replication material, here we only show the \code{predict()} call that produces the underlying data.

\begin{figure}[H]
    \centering
    \includegraphics[width=\linewidth]{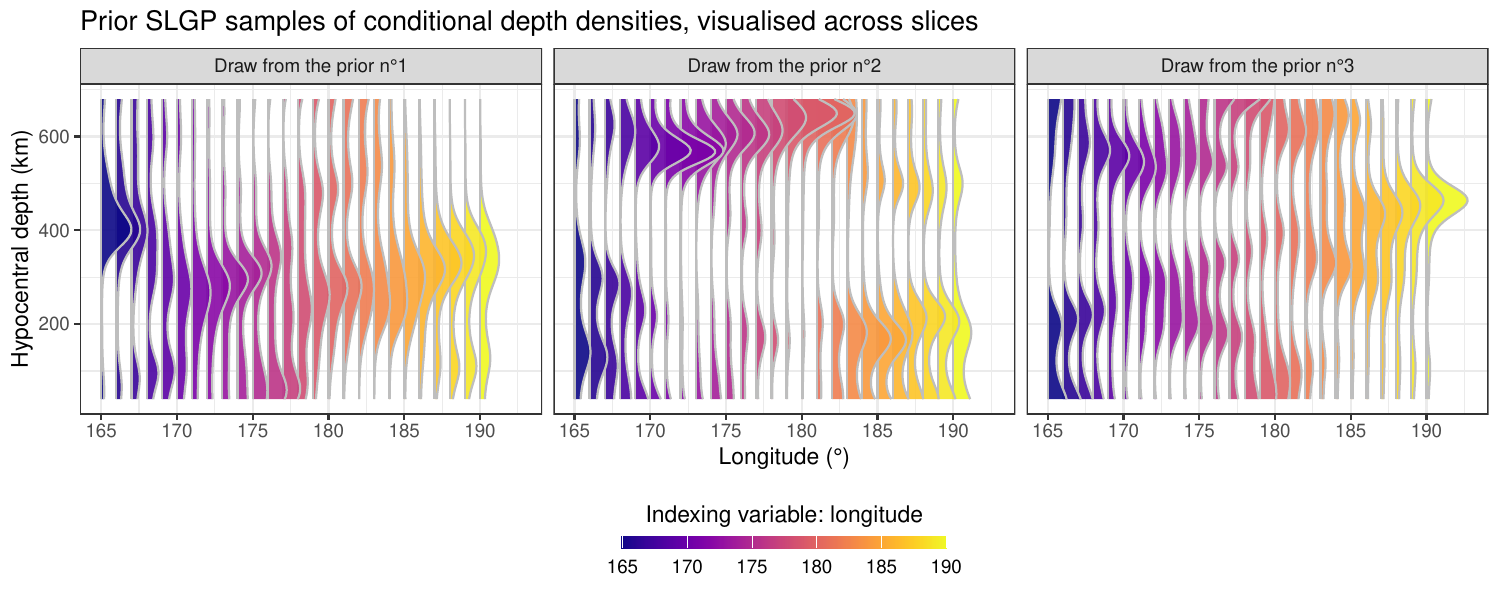}
\caption{Three draws from the SLGP prior using RFF with a Matérn 5/2 kernel. The data for this figure are generated with the \code{predict()} method, and the figure itself with \pkg{tidyr}, \pkg{ggplot2} and \pkg{viridis}.}
\label{fig:PriorSLGPggplot}
\end{figure}

Representing the pdfs along slices as we do allows for visualising the changes in shape and modalities in the generated field, by displaying pdfs generated for a SLGP for various values of $\xX$. The observed changes give encouraging insight into the SLGP's abilities to model a wide range of distributions. Other draws of SLGP's prior exhibit different changes in shape or modalities, illustrating the flexibility of our model.

We chose the Matérn 5/2 kernel to strike a balance between smoothness and variability. Other kernels induce different degrees of regularity, which influence the types of densities favoured by the SLGP. We compare three additional kernel choices: the exponential kernel, a special case of the Matérn class ($\nu = 1/2$) that yields less smooth realisations; the Matérn 3/2 kernel; and the Gaussian kernel, which leads to infinitely differentiable realisations. This is done through the following code, and the resulting realizations are illustrated in Figure~\ref{fig:OtherPriorSLGP}. We only show the code used to train the SLGPs, the resulting Figure~\ref{fig:OtherPriorSLGP} are obtained using the \texttt{predict()} method and advanced visual packages.

\begin{verbatim}
R> common_args <- list( formula = depth ~ long, data = df, method = "Prior",
+                    basisFunctionsUsed = "RFF", interpolateBasisFun = "WNN",
+                    hyperparams = list(lengthscale = c(0.15, 0.15), sigma2 = 1),
+                    sigmaEstimationMethod = "heuristic", opts=list(ndraws=1),
+                    predictorsLower = range_x[1], predictorsUpper = range_x[2],
+                    responseRange = range_response, seed = 1)
R> # One model per smoothness: Matérn 1/2 (exponential), 3/2 and Inf (Gaussian)
R> priors <- lapply(c(1/2, 3/2, Inf), function(nu)
+    do.call(slgp, c(common_args,
+                    list(opts_BasisFun = list(nFreq = 200, MatParam = nu)))))
\end{verbatim}

\begin{figure}[H]
    \centering
    \includegraphics[width=\linewidth]{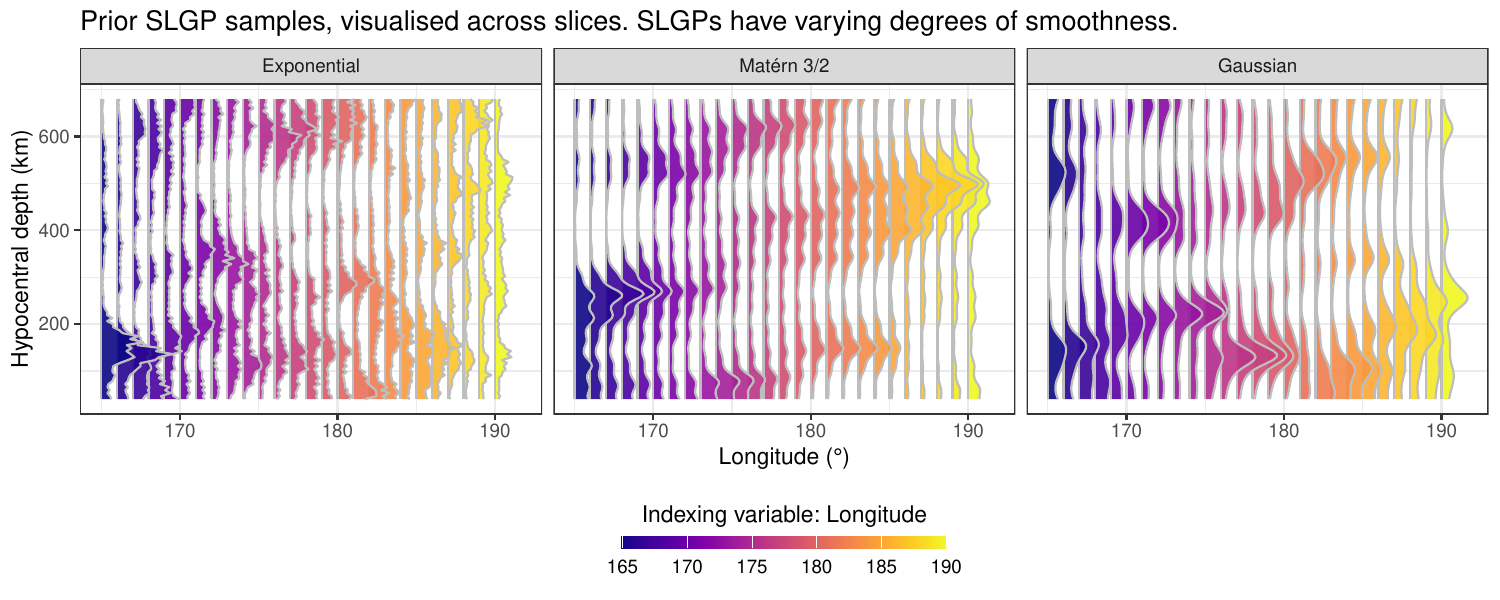}
\caption{Three draws of SLGPs induced by GPs with various kernels.}
\label{fig:OtherPriorSLGP}
\end{figure}
\vspace{-18pt}

\subsubsection{Maximum a Posteriori estimation }

Now that the prior has been defined, we can proceed with estimating the spatially varying density field. The fastest estimation scheme provided in our package is the MAP estimation procedure, it yields a single point estimate. However, this comes at the cost of limited inferential depth: MAP does not support uncertainty quantification, as it focuses solely on identifying the mode of the posterior distribution rather than characterizing its full shape.

The following code demonstrates how to update an existing SLGP model using the MAP estimation scheme. Alternatively, a MAP model can be fitted directly, using a call similar to that used to define \code{modelPrior}, with the data supplied and \code{method = "MAP"}.

\begin{verbatim}
R> modelMAP <- update(modelPrior, newdata = df, method = "MAP")
\end{verbatim}

The \texttt{summary()} method includes fit diagnostics, here the log-posterior at the mode.

\begin{verbatim}
R> summary(modelMAP)
Summary of a Spatial Logistic Gaussian Process (SLGP) model
    [[model description as above, not shown here for compactness]]


Estimation diagnostics
  Fitting time      : 0.94s   (setup 0.55s, estimation 0.39s)
  Coefficient draws : 1
  Log-posterior     : 655.461 (at mode for MAP/Laplace)
  Optimiser         : converged
\end{verbatim}

The fitted object has a basic \texttt{plot()} method. As with the prior, the output of \texttt{predict()} can be used instead to create a richer \pkg{ggplot2} figure (c.f. replication material).

\begin{figure}[H]
    \centering
    \includegraphics[width=0.95\linewidth]{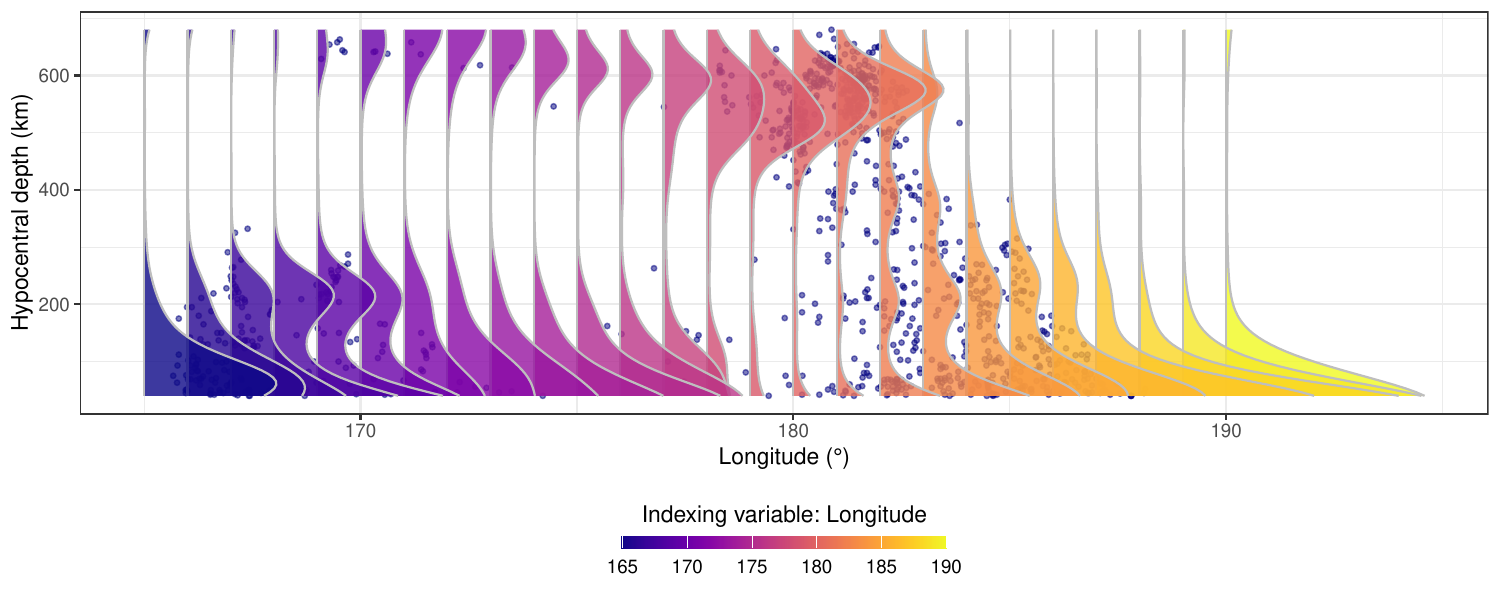}
\caption{MAP-based SLGP predictive densities of \depth across values of \longi.}
\label{fig:SLGPMAP}
\end{figure}
\vspace{-6pt}

Figure~\ref{fig:SLGPMAP} presents the predictive densities produced by the SLGP model, based on a heuristic length-scale set to 15\% of the input range. In Appendix~\ref{app:lengthscale_optim}, we explore a more principled approach by optimizing this length-scale through grid search, under an inverse gamma prior.

While MAP inference provides only a single estimate, its outputs remain interpretable and useful. Figure~\ref{fig:MAP_vs_hist} compares the MAP-predicted density slices to empirical histograms computed from the raw data.  Note that this comparison is provided for illustrative purposes only: one of the key motivations behind the development of SLGP modeling is to avoid the need for arbitrary binning of data prior to density estimation.

\vspace{-6pt}
\begin{figure}[H]
\centering
    \includegraphics[width=0.99\linewidth]{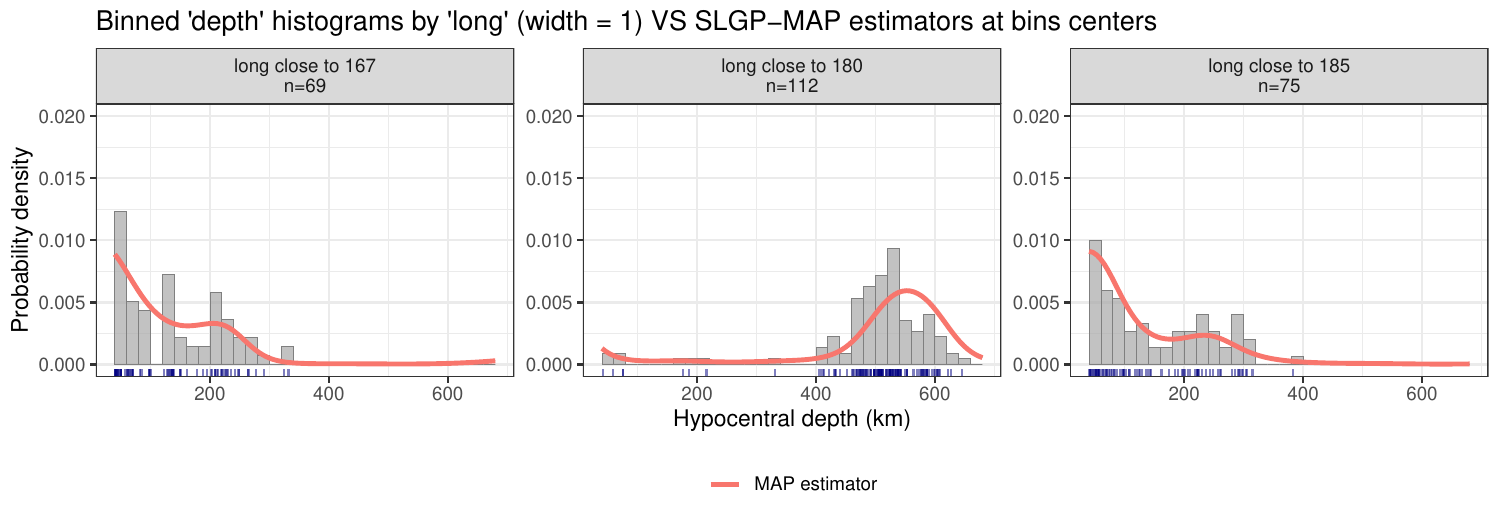}
\caption{SLGP MAP estimates of the conditional density of \depth versus empirical histograms, at three \longi slices.}
\label{fig:MAP_vs_hist}
\end{figure}
\vspace{-6pt}

It is nonetheless encouraging to observe that the SLGP model reproduces several key features of the data. This illustrates the effectiveness of SLGPs even under a fast, point-estimate-based inference strategy.

\subsubsection{Uncertainty quantification: Laplace approximation and Markov Chain Monte Carlo-based estimation}

To enable uncertainty quantification, we can integrate the MAP approach with Laplace approximation. It refines our estimation strategy by approximating the posterior distribution with a multivariate Gaussian. This method strikes a balance between the full Bayesian inference of MCMC and the computational efficiency of MAP estimation. By leveraging both the gradient and Hessian of the posterior, it captures essential curvature information, providing a more informed approximation of the posterior landscape. The following code demonstrates how to re-train our SLGP model, shows the resulting summary. The posterior mean density and ten posterior draws are displayed in Figure~\ref{fig:LAPLACE_draws}.

\vspace{-6pt}
\begin{verbatim}
R> modelLaplace <- update(modelMAP, newdata = df, method="Laplace")
R> summary(modelLaplace)
Summary of a Spatial Logistic Gaussian Process (SLGP) model
  [[model description as above, not shown here for compactness]]

Estimation diagnostics
  Fitting time      : 1.84s   (setup 0.51s, estimation 1.33s)
  Coefficient draws : 1000
  Log-posterior     : 655.461 (at mode for MAP/Laplace)
  Hessian nugget    : none (invertible)
  Hessian condition : 121.5
R> plot(modelLaplace, draw = c("mean", 1:10), panels = TRUE)
\end{verbatim}

\vspace{-6pt}
\begin{figure}[H]
\centering
    \includegraphics[width=\linewidth]{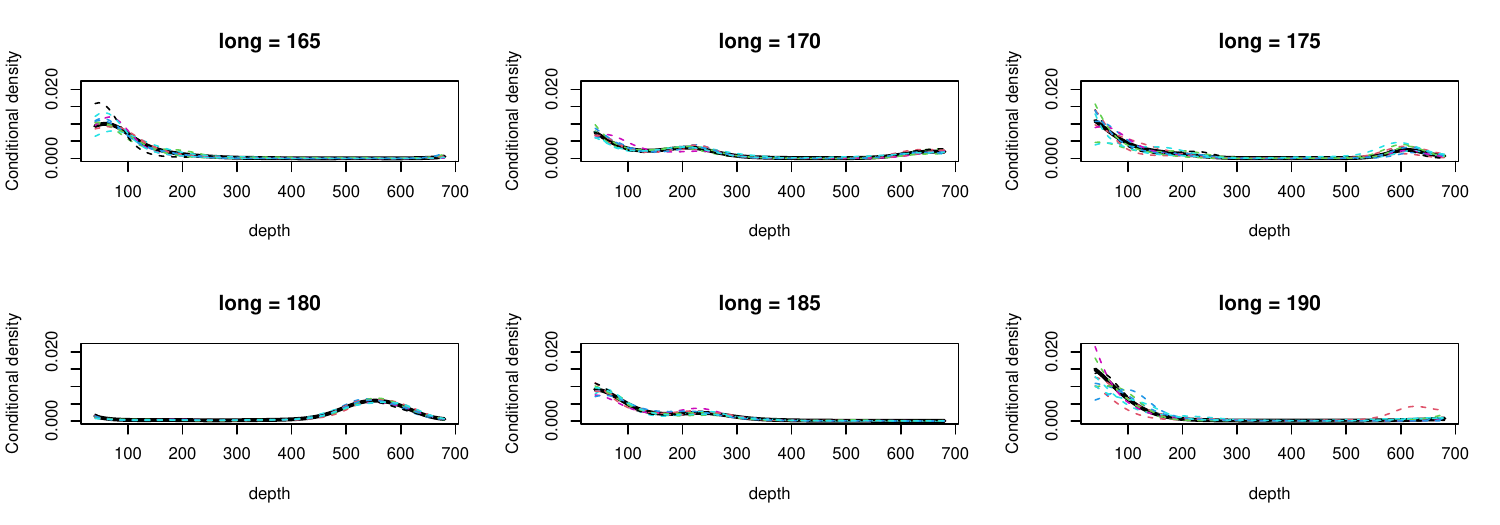}
\caption{Predictive conditional density of \depth at selected \longi values
under the Laplace approximation, with posterior draws superimposed.}
\label{fig:LAPLACE_draws}
\end{figure}

Using MCMC provides access to full Bayesian inference, offering posterior samples from the distribution of the SLGP model. This is the approach previously used in the paper to produce Figure~\ref{fig:teaser1}. The following code demonstrates how to re-train our SLGP model using the MCMC estimation-based scheme:
\begin{verbatim}
R> modelMCMC  <- update(modelMAP,  newdata = df,  method="MCMC",
+                       opts = list(stan_chains=2, stan_iter=1000), seed=1)
\end{verbatim}

MCMC provides a rich representation of uncertainty in the estimated density. The posterior mean and draws against the binned observations are those already shown in Figure~\ref{fig:teaser1a} of the introduction, obtained from the \code{predict()} method and \pkg{ggplot2}.

One of the advantages of the SLGP framework is that it predicts entire probability density functions (PDFs) over space. This opens the door to a wide range of nonlinear inferences on the estimated field. In particular, we can compute and visualize functionals of the predicted densities, such as moments (mean, variance, skewness, etc.) or quantiles.

In our current implementation, we provide predictions for both centered and uncentered moments of the estimated PDFs at each location. These derived quantities can themselves be interpreted as spatial fields and are informative summaries of the underlying random process. Importantly, when using probabilistic inference schemes like Laplace approximation or MCMC, the uncertainty in the SLGP predictions naturally propagates to these functionals, and we can also quantify uncertainty on the functionals. This is illustrated in Figure~\ref{fig:teaser1c} in the introduction, where we use the MCMC-trained SLGP model and the following code to compute the first four moments of the conditional depth distribution along the transect.

\begin{verbatim}
R> dfX <- data.frame(long=seq(range_x[1], range_x[2], 1))
R> # Uncentered moments : the mean
R> predMean <- predict(modelMCMC, type= "moments", newdata = dfX, 
+                      power=c(1), centered=FALSE) 
R> # Centered moments: variance, Kurtosis and Skewness
R> predVar <- predict(modelMCMC, type= "moments", newdata = dfX, 
+                     power=c(2, 3, 4), centered=TRUE) 
\end{verbatim}

We also support the joint prediction of quantiles at arbitrary levels. We display the code used to generate the quantiles and show the resulting plot in Figure~\ref{fig:teaser1b}. In this figure, we conduct a sanity check by comparing SLGP-based quantile estimates with those obtained using classical quantile regression \citep{koenker2005quantile}, as implemented in the \proglang{R} package \pkg{quantreg} \citep{quantreg}, and also empirical quantiles on a rolling window.

\begin{verbatim}
R> pred <- predict(modelMCMC, type = "quantiles", newdata = dfX,
+                  probs = c(5, 25, 50, 75, 95)/100)
\end{verbatim}

The three approaches mostly agree in the data-rich part of the domain, but differences arise where the data are sparse or the conditional distribution is bimodal: there, the SLGP quantiles track the rolling-window empirical ones more closely. The \pkg{quantreg} curves also cross one another in places, as any method estimating quantiles one level at a time without an underlying density may. Being read off the estimated density, the SLGP quantiles are non-crossing by construction.

Everything so far used a continuous response and a scalar index, a setting chosen only because it makes slice-based figures readable. The next two subsections relax each in turn: Section~\ref{subsec:discrete} takes a discrete response and shows how prior knowledge about the observation process enters through a trend, and Section~\ref{subsec:2Dfield} indexes by the full epicentre location.

\subsection{Discrete probability estimation with a SLGP model}
\label{subsec:discrete}

The SLGP is not restricted to continuous responses. The normalising integral is
computed by quadrature, and aligning the quadrature nodes with a discrete support turns the model into an estimator of conditional probability mass functions. The normalisation is then exact rather than approximate, since the integral reduces to a finite sum. The discrete-output formulation, its relation to multinomial logistic regression and to multiclass GP classification, and its use in sequential design are developed in \citet{gautier_discrete_2026}, here we simply show how the package
handles it. The \code{quakes} dataset contains a natural discrete response: the event magnitude \magn, which is reported on the regular grid $\{4.0, 4.1, \dots, 6.4\}$. We display the corresponding data in Figure~\ref{fig:QuakesScatterDiscrete}

\begin{figure}[H]
\centering
    \includegraphics[width=0.95\linewidth]{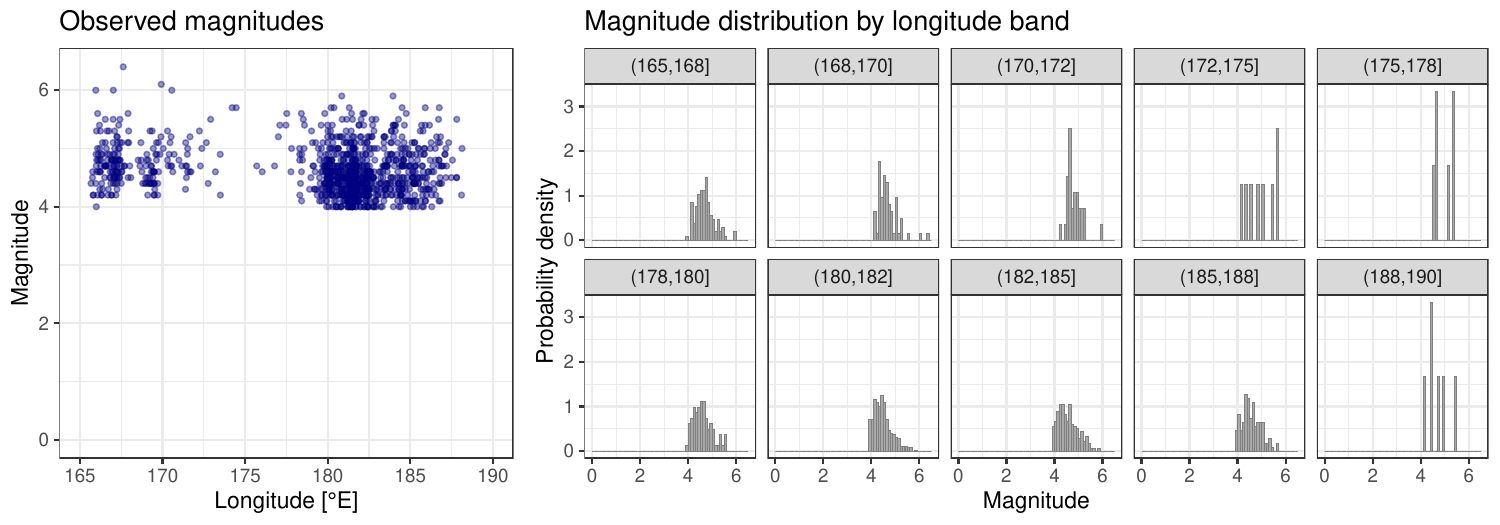}
\caption{Event magnitude \magn against \longi in the \code{quakes} catalogue (left; the dashed line marks the completeness threshold of $4.0$), and empirical magnitude distributions within longitude bands (right).}
\label{fig:QuakesScatterDiscrete}
\end{figure}

This variable also illustrates a second point. The dataset only records events of magnitude at least \(4.0\): smaller events are not reliably detected by the seismic network, so they are absent from the data. The observed distribution of \texttt{mag} is thus the distribution of magnitude conditional on exceeding the completeness threshold of \(4.0\). We use this as an opportunity to demonstrate the SLGP trend: rather than simply restricting the response support to \([4.0, 6.5]\), we model the full range \([0, 6.5]\) and encode the detection mechanism as prior knowledge through a trend function that suppresses probability mass below the threshold.

The trend enters the SLGP as an additive term on the latent Gaussian process, $Z_{\xX, t} = m(\xX, t) + \sigma \sum_j f_j(\xX, t) \varepsilon_j$, so that it shifts the prior mean of the log-density: before any data are seen, the prior expects the density to be proportional to $e^{m(\xX, t)}$. Here a trend that is constant in $\xX$ and takes a large negative value below the threshold suffices,
\begin{equation*}
m(\magn) = -10\,\times\mathbf{1}_{\{\magn < 4.0\}},
\end{equation*}
so that it suppresses mass below the threshold by four to five orders of
magnitude without forbidding it outright, and leaves the density unconstrained
above it. The same device is used in \citet{gautier_discrete_2026} to encode a
varying support.

The trend is supplied to \code{slgp()} as a function of a data frame containing the response and covariate columns. One implementation detail matters: the trend is evaluated on the response in its original units, so the breakpoint is the physical value $4.0$ rather than its image in $[0, 1]$.

\begin{verbatim}
R> trend_fun <- function(df){return(ifelse(df$mag < 4, -10, 0))}
R> modelMAGdisc <- slgp(mag~long, 
+                     data=df, method="MAP", discrete=TRUE, nIntegral = 66,
+                     basisFunctionsUsed = "RFF", interpolateBasisFun="WNN", 
+                     hyperparams = list(lengthscale=c(0.15, 0.15), sigma2=1), 
+                     sigmaEstimationMethod = "heuristic", # rewrites sigma2
+                     predictorsLower= range_x[1], predictorsUpper= range_x[2],
+                     responseRange= c(0, 6.5), trend=trend_fun, seed=1,
+                     opts_BasisFun = list(nFreq=250, MatParam=5/2))
R>  plot(modelMAGdisc, draw = "mean", panels = TRUE)
\end{verbatim}
\vspace{-6pt}

The resulting conditional probability mass function is shown in Figure~\ref{fig:HistMAPDiscrete1}. 
Figure~\ref{fig:HistMAPDiscrete2} gives another visualisation at three longitudes, assembled with \pkg{ggplot2} from the output of \code{predict()}.

\vspace{-6pt}
\begin{figure}[H]
\centering
    \includegraphics[width=0.9\linewidth]{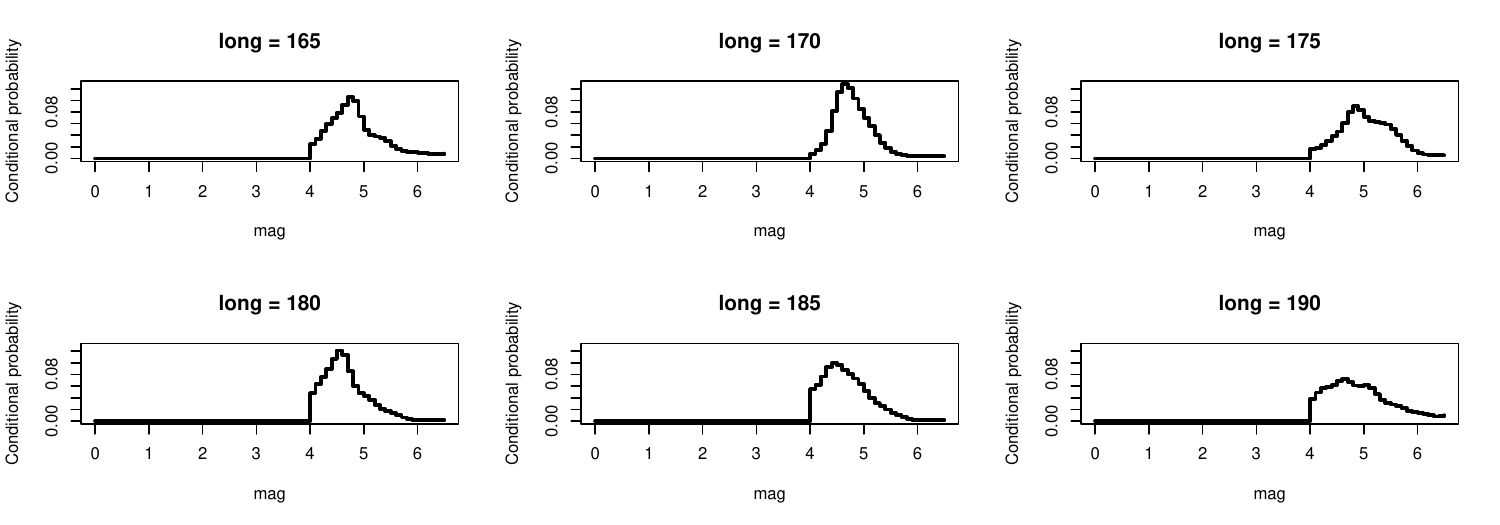}
\caption{Fitted conditional probability mass function of \magn at six longitudes, from the base \code{plot()} method. The trend suppresses probability mass below $4.0$.}
\label{fig:HistMAPDiscrete1}
\end{figure}

\vspace{-6pt}
\begin{figure}[H]
\centering
    \includegraphics[width=0.925\linewidth]{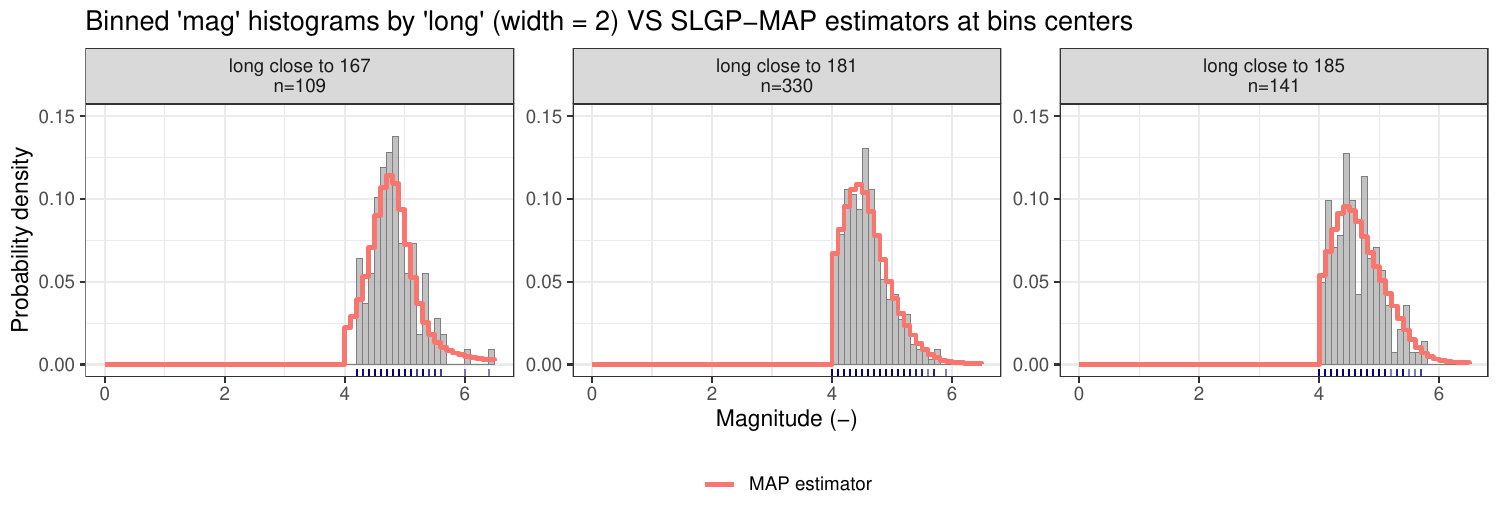}
\caption{SLGP-MAP conditional pmf of \magn (step curve) versus the empirical magnitude frequencies, at three \longi slices.}
\label{fig:HistMAPDiscrete2}
\end{figure}
\vspace{-6pt}

Two things are worth retaining from this example. First, a discrete response requires no separate model: setting \code{discrete = TRUE} and aligning \code{nIntegral} with the size of the support turns the normalising integral into a finite sum. Second, a trend is a convenient way of encoding what is known about the observation process. Here it separates cleanly what the data can speak to (magnitudes above the completeness threshold) from what they cannot (magnitudes below it), instead of silently conflating the two by truncating the support.

\subsection{A bi-dimensional index and continuous response application}
\label{subsec:2Dfield}

While the one-dimensional setting is convenient for slice-based figures, it is also within reach of many conditional density estimators. The distinctive strength of the SLGP is that nothing in its construction is tied to a scalar index. We now exploit the full epicentre location \texttt{(long,\ lat)} as a two-dimensional index and estimate the conditional distribution of \texttt{depth} over the region. This is a regime where per-cell histograms become untenable, since binning a two-dimensional index with 1000 scattered events leaves most cells empty, yet the SLGP applies unchanged.

To obtain an honest assessment of the fit rather than a purely visual one, we adopt a train/test split: the model is trained on a random 75\% of the events and validated on the remaining 25\%. The left panel of Figure~\ref{fig:figureQuakes2DMAPcenter} shows the two subsets over the region.

Fitting the two-dimensional model requires only that the index-related arguments grow accordingly: the formula becomes \texttt{depth\ \textasciitilde{}\ long\ +\ lat}, \texttt{predictorsLower} and \texttt{predictorsUpper} are length two, and the length-scale vector has one entry per index dimension plus one for the response. The estimation scheme, basis and quadrature are otherwise identical to the one-dimensional case.

\vspace{-6pt}
\begin{verbatim}
R> set.seed(1); id_train <- sample(seq_len(nrow(df)), size = 750)
R> df_train <- df[id_train, ]; df_test  <- df[-id_train, ]
R> modelMAP2D <- slgp(depth ~ long + lat,
+                     data = df_train, method = "MAP",
+                     basisFunctionsUsed = "RFF", interpolateBasisFun = "WNN",
+                     sigmaEstimationMethod = "heuristic",  seed = 1,
+                     predictorsLower = c(165,-40), predictorsUpper = c(190, -10),
+                     responseRange = range_response,
+                     opts_BasisFun = list(nFreq = 250, MatParam = 5/2))
R> summary(modelMAP2D)
Summary of a Spatial Logistic Gaussian Process (SLGP) model
  Formula        : depth ~ long + lat
  Response       : depth  in [40, 680]
  Covariate(s)   : long, lat
  Estimation     : MAP
  Basis          : RFF  (rank p = 500)
  Length-scales  : 0.15, 0.15, 0.15
  Variance sigma2: 0.6967
  Data           : 750 obs, 748 distinct covariate value(s)

Estimation diagnostics
  Fitting time      : 14.9s   (setup 9.67s, estimation 5.28s)
  Coefficient draws : 1
  Log-posterior     : 634.894 (at mode for MAP/Laplace)
  Optimiser         : converged
\end{verbatim}

Before inspecting the fitted field, we briefly check that it is calibrated on data it has never seen. Because the SLGP returns a full predictive distribution at every location, we can apply the probability integral transform (PIT): for each held-out event we evaluate the predictive CDF at the observed depth. This is also a first use of the \code{type = "cdf"} mode of \code{predict()}, which returns $F(t \mid \xX)$ on any set of (index, response) pairs rather than the density. If the model is well calibrated, these PIT values should be approximately uniform on $[0, 1]$.

We assess this both visually, with the PIT histogram and uniform Q-Q plot of Figure~\ref{fig:SLGPfitting2DPIT}, and with two complementary goodness-of-fit tests against the uniform. The Kolmogorov-Smirnov statistic is the largest vertical gap between the empirical and uniform CDFs, and so reacts to a localised departure, while the Cramér-von Mises statistic integrates the squared gap over \([0,1]\), and so reacts to a diffuse, systematic one. 

\begin{verbatim}
R> pit_pred <- predict(modelMAP2D, type = "cdf", newdata = df_test)
R> ks.test(pit_pred$cdf_1, "punif", 0, 1) # D = 0.0625, p = 0.282
R> library(goftest); cvm.test(pit_pred$cdf_1) # omega2 = 0.3467, p = 0.100
\end{verbatim}

\vspace{-6pt}
\begin{figure}[H]
\centering
    \includegraphics[width=0.75\linewidth]{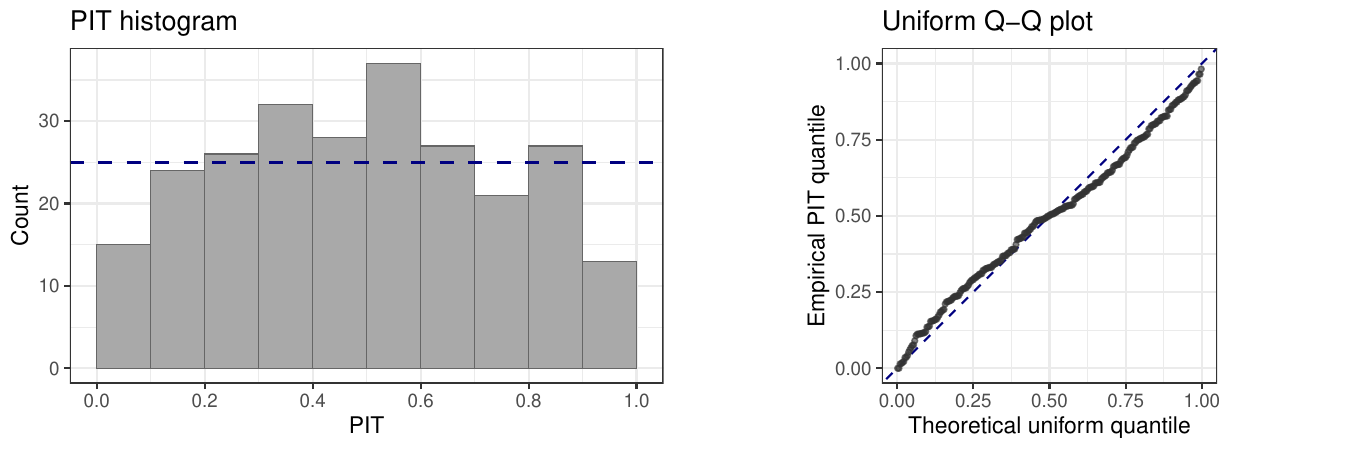}
\caption{Left: PIT histogram of the held-out depths (dashed line: expected count under uniformity). Right: uniform Q-Q plot of the PIT values.}
\label{fig:SLGPfitting2DPIT}
\end{figure}
\vspace{-6pt}

Both tests fail to reject uniformity at 10\% level, the PIT histogram is reasonably close to flat and the points of the Q-Q plot lie near the diagonal, confirming that the fitted field is reasonably calibrated on held-out data. We complement these assessment with a visual check. We select six rectangles, pool the observations falling inside each of them into a histogram, and show the SLGP conditional density evaluated at the center of the rectangles (Figure~\ref{fig:figureQuakes2DMAPcenter}).

\begin{figure}[H]
\centering
    \includegraphics[width=\linewidth]{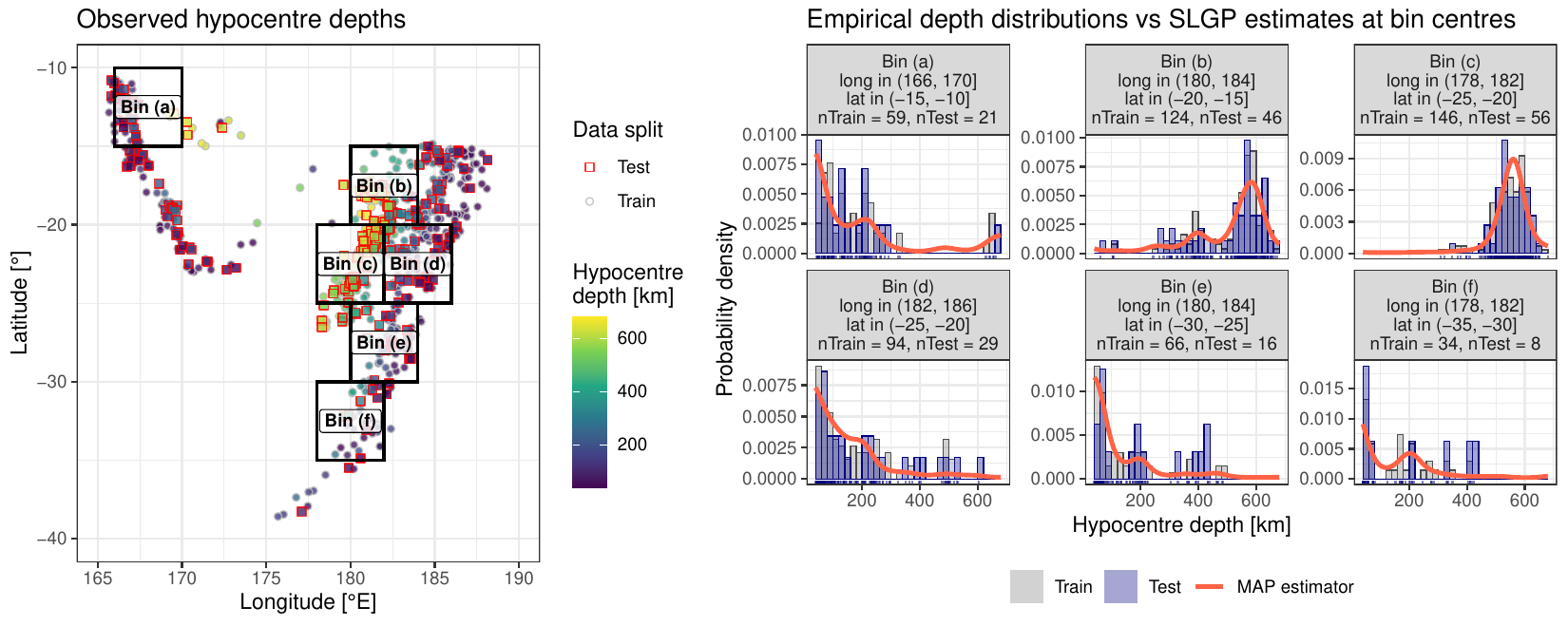}
\caption{Empirical depth distributions (grey) versus the SLGP conditional density evaluated at the centre of each rectangle (red).}\label{fig:figureQuakes2DMAPcenter}
\end{figure}
\vspace{-6pt}

The agreement is good in the well-sampled rectangles (a)-(c) and looser in (d)-(f). Part of this discrepancy is an artefact of the comparison rather than a failure of the model: the histogram pools every observation inside a fairly large and arbitrary rectangle, whereas the SLGP curve is evaluated at the single point at its centre. If the conditional distribution varies appreciably across the rectangle, the two objects need not coincide even for a perfectly specified model.

Since the SLGP estimates the entire conditional distribution at each location, the probability \(P(\text{depth} > u \mid \text{long}, \text{lat})\) of an event deeper than a threshold \(u\) is read directly from the predictive CDF, for any threshold. Figure~\ref{fig:teaser2} maps this probability at \(u = 300\,\)km, isolating the deep Wadati-Benioff seismicity.

\begin{verbatim}
R> grid2D_cdf <- expand.grid(long=seq(165, 190), lat=seq(-40, -10), depth=300)
R> pred_cdf <- predict(modelMAP2D, type = "cdf", newdata = grid2D_cdf)
R> grid2D_cdf$exceed <- 1 - pred_cdf$cdf_1
\end{verbatim}

This is how Figure~\ref{fig:teaser2}, shown in the introduction, is produced. This single map summarises a genuinely spatial, distribution-valued object: at every location it reports not a point prediction but a probability derived from the full conditional law, and it does so continuously across the region, including where events are sparse. We stress the reading is most reliable where the epicentres are dense: in the poorly sampled corners of the domain the field is an extrapolation of the smooth GP prior and should be interpreted with the corresponding caution.

Finally, the fitted object is a full generative model, not only a density evaluator. \texttt{simulate()} draws synthetic responses at chosen locations by inverse-transform sampling from the predictive CDF, which is useful for propagating distributional uncertainty into downstream computations.

\begin{verbatim}
R> sims <- simulate(modelMAP2D, nsim = 1, newdata = df[, c("long", "lat")])
\end{verbatim}

Figure~\ref{fig:simulation2D} shows one such synthetic catalogue, drawn at the epicentres of the observed events. It reproduces the main structure of the region (shallow events along the eastern arc, deep events concentrated in the western band) without reproducing any individual observation.

\vspace{-6pt}
\begin{figure}[H]
\centering
    \includegraphics[width=0.85\linewidth]{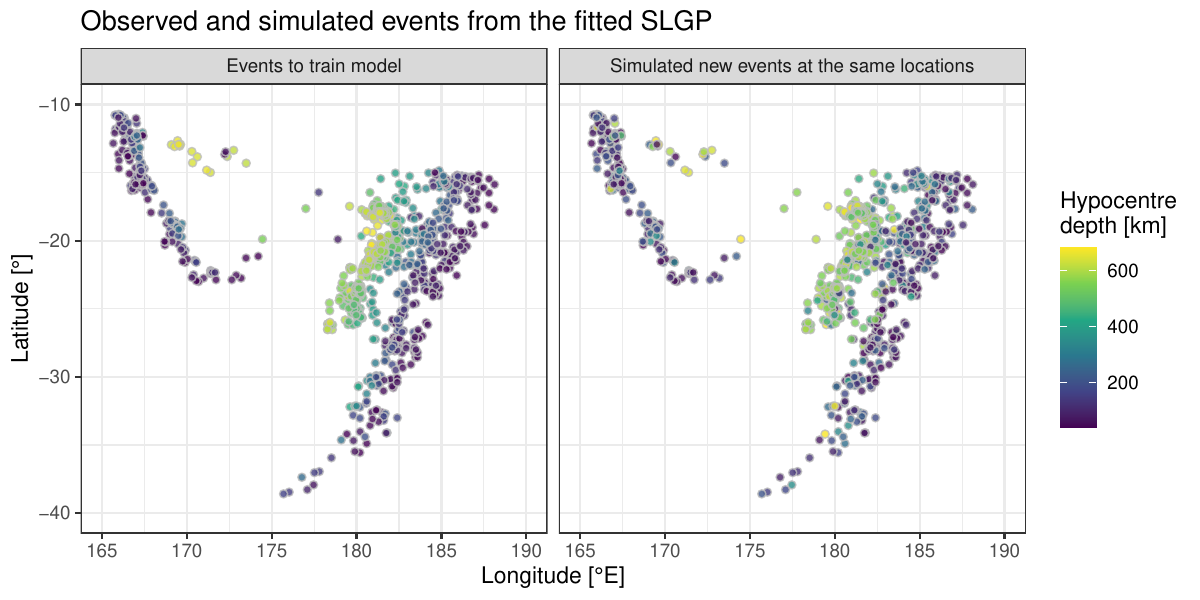}
\caption{Observed earthquake depths (left) and one realization simulated from the fitted SLGP at the same spatial locations (right).}
\label{fig:simulation2D}
\end{figure}

\section{A numerical study of estimation quality and cost}
\label{section:benchmark}

To assess the quality of SLGP-based density field estimation, we go down to the 1D continuous case and define a reference field using the MAP estimate, which serves as the ground truth. We then generate synthetic data from this reference field and compare different estimation setups to evaluate their impact on accuracy. We explore three main experimental factors, and for each setting we replicate the experiment 25 times with several seeds to account for variability:

\begin{itemize}
    \item \emph{Number of basis functions.} This experiment evaluates the impact of the number of Random Fourier Features. A larger number of basis functions is expected to improve accuracy but to increase computational cost. This setup allows us to study the approximation error induced by using a finite rank and to assess the trade-off between speed and expressiveness.

    \item \emph{Length-scale misspecification.} Here we test the sensitivity of SLGP estimation to the specification of the length-scale hyper-parameter. We train models with length-scales that are too small (0.01, 0.05), well specified (0.15), or too large (0.25, 0.5). This analysis highlights the consequences of hyper-parameter selection.

    \item \emph{Sampling scheme.} This experiment investigates how the spatial structure of the data affects the quality of estimation. We vary the way the index \longi is sampled: uniformly at random over the domain; restricted to a regular grid of 5, 11 or 21 longitudes, so that observations pile up on a few locations; and uniformly at random with two excluded longitude bands, mimicking corridors left uncovered by the seismic network. These variations reveal the method's ability to generalise in poorly observed regions.
\end{itemize}

This analysis provides practical insights into how SLGP estimation behaves under different conditions, guiding model selection for real applications.

For each tested configuration, we compute the estimation error, comparing the inferred density field to the reference field. This allows us to measure the effect of different modelling choices. To assess the similarity between estimated and true distributions, we compute several integrated distances based on both the probability density functions (PDFs) and the cumulative distribution functions (CDFs). 

The considered distances are defined for a distribution $P_1$ that admits a pdf $f_1$ and cdf $F_1$ and a distribution $P_2$ that admits a pdf $f_2$ and cdf $F_2$ as:
\begin{itemize}
    \item Hellinger Distance $d_H(f_1, f_2) = \sqrt{\frac{1}{2} \int \left( \sqrt{f_1(t)} - \sqrt{f_2(t)} \right)^2 dt}$
    \item Total Variation Distance $TV(f_1, f_2) = \frac{1}{2} \int \left| f_1(t) - f_2(t) \right| \,dt$ 
    \item Kullback–Leibler Divergence $KL(f_1, f_2) = \int f_1(t) \log \left( \frac{f_1(t)}{f_2(t)} \right) dt $
    \item Cramér–von Mises Distance:  $CvM(F_1, F_2) = \int \left( F_1(t) - F_2(t) \right)^2 f_1(t) \, dt$
\end{itemize}

These metrics are evaluated across a grid of values and provide different perspectives on the quality of the density estimation.

\subsection{Impact of the number of basis functions}

We first evaluate the effect of the number of RFF on SLGP estimation, keeping the length-scale fixed at the well-specified value and the design strategy as uniform. For each configuration, we generate synthetic datasets from the reference field and fit SLGP models using MAP estimation. The resulting estimated densities are then compared to the reference field by integrating the distances defined previously and the results are displayed in Figure~\ref{fig:benchmark1}.

\begin{figure}[H]
\centering
    \includegraphics[width=\linewidth]{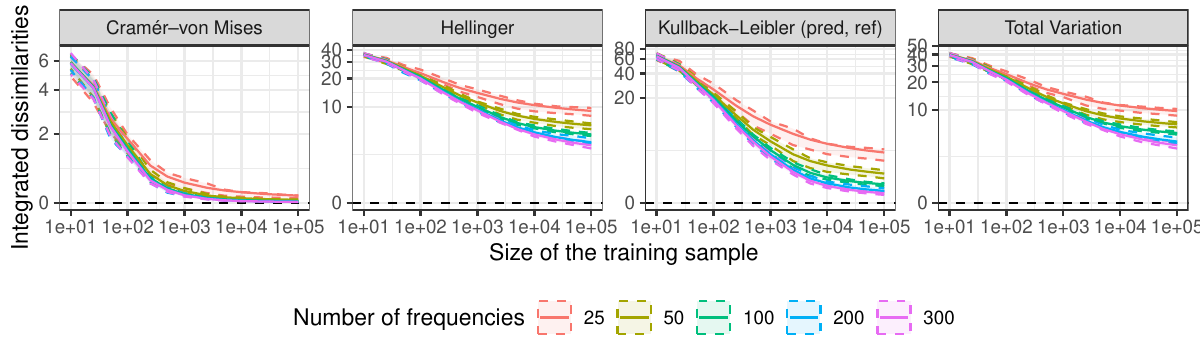}
\caption{Effect of the number of basis functions on SLGP density estimation. Integrated distances between the MAP estimates and the reference field are shown across varying training sample sizes: mean (solid lines) and 25th–75th percentile range (shaded ribbon).}\label{fig:benchmark1}
\end{figure}

As expected, increasing the number of basis functions improves estimation accuracy. Beyond a threshold ($n_{\text{Freq}} = 200$), additional features provide negligible gains. Also, the Kullback–Leibler divergence and Cramér–von Mises distance converge rapidly toward zero, while the Hellinger distance and Total Variation distance improve more gradually. This is due to sensitivities of these metrics to local versus global discrepancies between densities.

\subsection{Impact of the length-scale selection}

Using the well-specified number of basis functions ($200$) and a uniform sampling of \longi, we next investigate how the choice of length-scale affects MAP-based SLGP estimation, with the results being displayed in Figure~\ref{fig:benchmark2}.

\begin{figure}[H]
\centering
    \includegraphics[width=\linewidth]{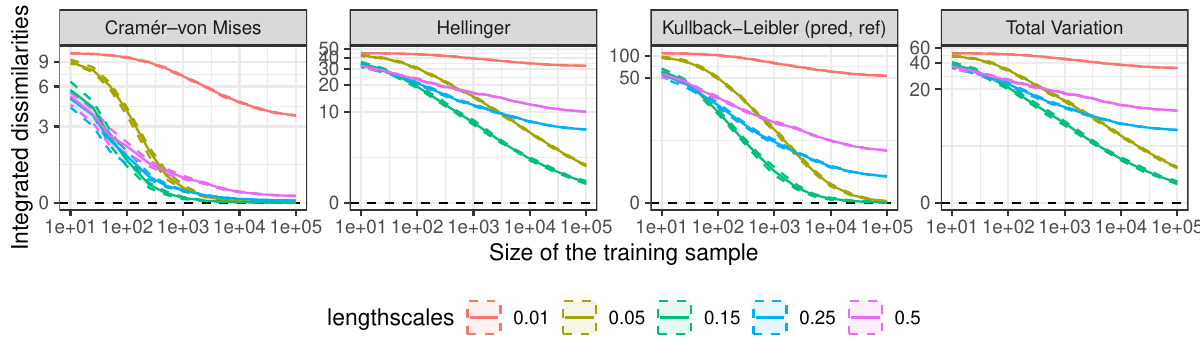}
\caption{Effect of the lengthscale on SLGP density estimation using MAP. Integrated distances between the MAP estimates and the reference field are shown across varying training sample sizes: mean (solid lines) and 25th–75th percentile range (shaded ribbon).}
\label{fig:benchmark2}
\end{figure}

The analysis indicates that the best performance across all metrics is obtained using the true or near-true lengthscale. Short lengthscales lead to overly localized predictions, failing to capture broader trends, whereas long lengthscales oversmooth the distributions, missing fine-scale features such as multi-modality. The well-specified lengthscale provides a balanced compromise, reproducing both global and local structure in the reference field.

\subsection{Impact of the Sampling Scheme}

Finally, we assess the influence of the sampling design on SLGP performance. We fix the number of basis functions and the length-scale to well-specified values and compare MAP estimates to the reference field across different \longi selection schemes:
\begin{itemize}
\item uniform on a regular grid of 5 longitudes: $\{165, 171.25, 177.5, 183.75, 190\}$,
\item uniform on a regular grid of 11 longitudes: $\{165, 167.5, \dots, 190\}$,
\item uniform on a regular grid of 21 longitudes: $\{165, 166.25, \dots, 190\}$,
\item uniform on the full segment $[165, 190]$,
\item and uniform with holes, i.e. uniform on $[165, 170] \cup [175, 184] \cup [186, 190]$.
\end{itemize}

We visualise the integrated dissimilarities for varying sample sizes in Figure~\ref{fig:benchmark3}.

\begin{figure}[H]
\centering
    \includegraphics[width=\linewidth]{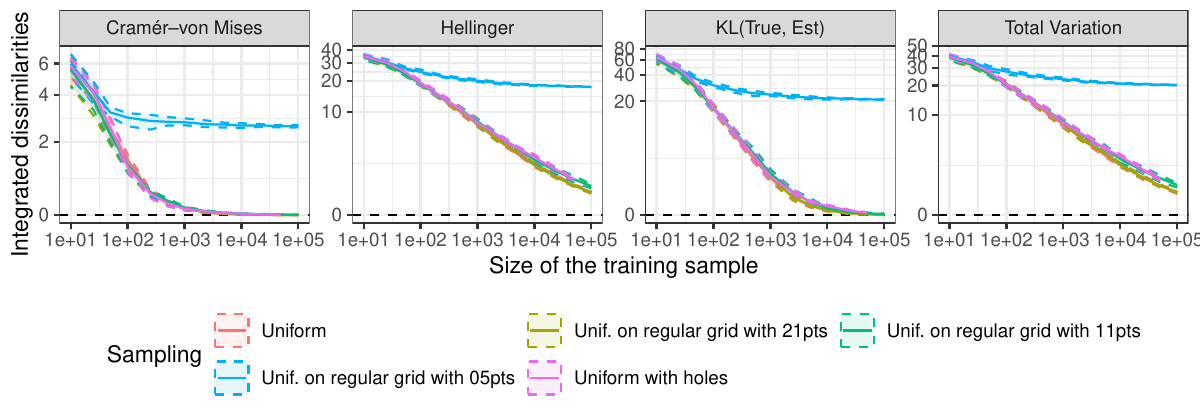}
\caption{Impact of selection scheme for \longi in SLGP density estimation using MAP. Integrated distances between the MAP estimates and the reference field are shown across varying training sample sizes: mean (solid lines) and 25th–75th percentile range (shaded ribbon).}\label{fig:benchmark3}
\end{figure}

Estimation performance depends on the sampling scheme used for \longi. The ``Uniform on a regular grid with 5 points'' design consistently leads to the weakest performance, characterized by a systematic estimation bias. In comparison, the remaining designs achieve similar levels of accuracy, with the ``Uniform with holes'' scheme showing only a modest loss in efficiency.

To better understand these differences, we can look at localised errors in addition to the integrated ones, and examine how the estimation error varies as a function of \longi. We display in Appendix~\ref{app:figures} the Kullback-Leibler divergence, the Hellinger distance, the Cramér-von Mises criterion and the total variation distance between the estimated probability density function at a given \longi and the corresponding reference density in Figures~\ref{fig:benchmark3a}, \ref{fig:benchmark3b}, \ref{fig:benchmark3c} and \ref{fig:benchmark3d}. 

These results illustrate the impact of localisation: regions with missing or sparse data exhibit higher errors, and dense coverage in some areas does not always compensate for missing information elsewhere. This analysis suggests a practical guideline: if the goal is to learn the full field rather than targeting specific locations, it is preferable to distribute sampling points evenly across the space rather than accumulating them at selected covariates. However if the goal is to refine knowledge at prescribed location, intensifying sampling in the area leads to better local predictions. Finally, these observations raise the question of the interplay between sampling density and the length-scales of the underlying process. Finally, these observations raise the question of the interplay between sampling density and the length-scales of the underlying process, which we return to in the conclusion.

\subsection{Computational cost}
\label{subsec:cost}
The experiments above measure how well the SLGP estimates the field. This one
measures what it costs. A fit has two phases, which \code{timing()} reports
separately: \emph{setup}, covering normalisation, the quadrature
pre-computation and the evaluation of the basis functions at the nodes, and
\emph{estimation}, covering the optimisation or the sampling itself. We cross
the sample size $n \in \{10^2, 10^3, 10^4\}$ and the rank
$p \in \{50, 100, 200, 500\}$ with the integral approximation scheme, at
\code{method = "MAP"}, and with the estimation method, at
\code{interpolateBasisFun = "WNN"}. Each configuration is replicated ten times,
with a different sample and a different draw of frequencies. The total cost of a MAP fit under the three schemes is shown in Figure~\ref{fig:timingScheme}.
\begin{figure}[H]
\centering
    \includegraphics[width=\linewidth]{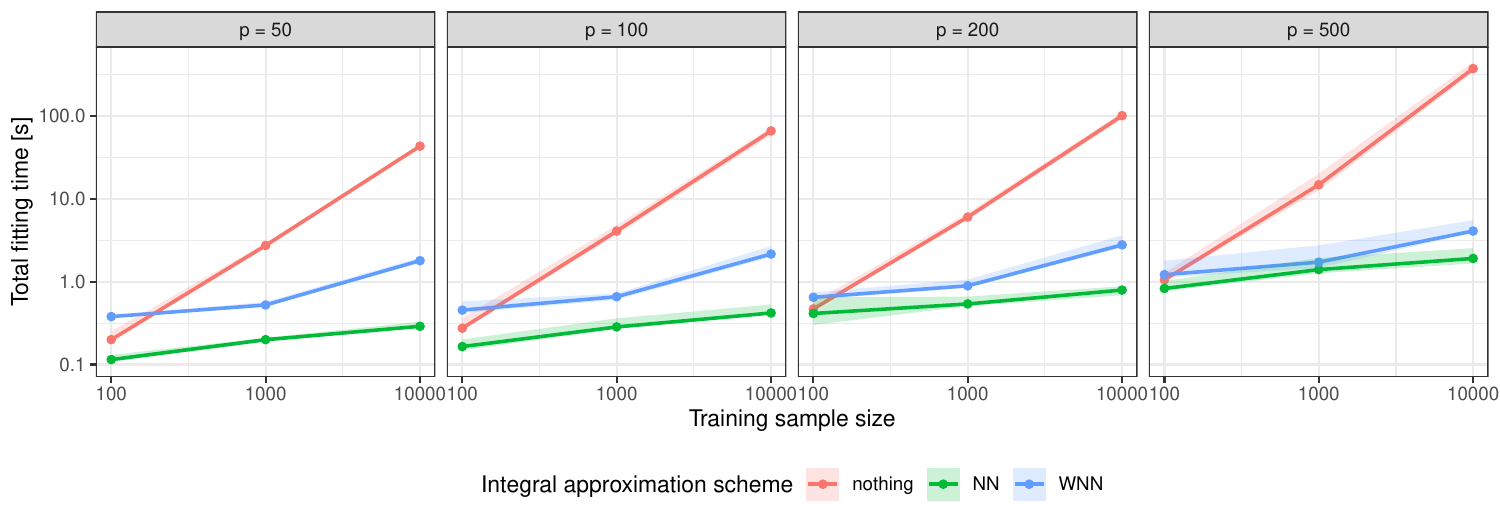}
\caption{Total time of a MAP fit against the sample size, for the three integral approximation schemes and four ranks. Lines are medians over ten replicates, ribbons the 10th-90th percentiles. Both axes are logarithmic.}
\label{fig:timingScheme}
\end{figure}
The exact scheme evaluates one normalising integral per distinct covariate
value, and under uniform sampling there are as many of these as observations. At
$n = 10^4$ a MAP fit takes between 44\,s ($p = 50$) and 376\,s ($p = 500$) with
the exact scheme, against 0.3 to 1.9\,s for \code{NN} and 1.8 to 4.1\,s for
\code{WNN}. The saving is not confined to the pre-computation: the matrix of
basis values the optimiser works with has one block of quadrature nodes per
distinct covariate value under the exact scheme, but at most \code{nDiscret}
blocks under \code{NN} and \code{WNN}, so every gradient evaluation is cheaper
in the same proportion.

The two interpolated schemes are not equivalent: \code{WNN} is 1.2
to 6 times slower than \code{NN}, because it combines every observation with the vertices of the grid cell that contains it, a step that grows with $n$, whereas \code{NN} aggregates observations through multiplicities. We nevertheless recommend \code{WNN}, whose interpolation is what keeps the approximate likelihood smooth in the covariate.

Turning to the estimation methods, at $p = 500$ the estimation phase of Laplace
costs 4.7 to 6.5 times that of MAP, the extra work being a single evaluation of
the Hessian at the mode in closed form. MCMC costs 2500 to 7500
times as much, that is from 11 minutes ($n = 10^2$) to 3.4 hours ($n = 10^4$)
for four chains of 2000 iterations run one after the other, against 1.5 to
11\,s for Laplace. Of course, running the chains in parallel divides the wall-clock cost by close to the number of chains. 

Additional views of these experiments (the share of the total time spent in estimation, and the growth of the estimation time with the rank) are given in Appendix~\ref{app:timing}.

\section{Conclusion and perspective}
In this work, we provided a practical introduction to Spatial Logistic Gaussian Process (SLGP) modeling, including a tutorial for implementation in \texttt{R}. We showed that SLGPs offer a flexible and expressive framework for modeling spatially dependent probability distributions, capable of capturing heterogeneity and multimodality while supporting both point estimation, via maximum a posteriori inference, and fully probabilistic inference, for instance through MCMC or Laplace approximations. Beyond accurate density estimation, SLGPs naturally yield a range of informative distributional summaries, including quantiles, moments, and other functionals of interest. Moreover, the framework extends seamlessly to discrete and ordinal outcomes, thereby broadening its applicability well beyond the continuous setting. All accompanying code is publicly available, ensuring full reproducibility and facilitating adoption in a wide range of applied contexts \citep{gautier_athenais_github_imple}.

Through systematic experiments, we illustrated the practical aspects of SLGP modeling, emphasizing the influence of hyperparameter choices (e.g., length-scale, number of basis functions) and sampling design on estimation quality. These insights provide concrete guidance for configuring SLGP models and for planning data collection strategies, particularly in contexts where observations are costly or unevenly distributed.

Looking forward, SLGPs offer substantial promise for advanced statistical inference tasks, including uncertainty quantification, goal-oriented modeling, and adaptive experimental design. In particular, extensions to stochastic inversion using SLGP-based Approximate Bayesian Computation (ABC) \citep{gautier_modelling_2023} provide a flexible approach to tackle high-dimensional, stochastic inverse problems. This framework opens avenues for principled, data-driven learning for complex systems.

Several directions for future research naturally follow from this work. A first important avenue concerns the interplay between sampling density and the effective length-scales of the SLGP. While our numerical experiments highlight the influence of design choices on estimation quality, a systematic theoretical and empirical investigation of this trade-off—particularly in the context of adaptive experimental design—remains to be conducted. Understanding how sampling density should scale with the smoothness of the underlying distributional field would provide principled guidance for data acquisition and design strategies.

Moreover, the SLGP framework is well suited to adaptive experimental design, as it yields full probabilistic descriptions of conditional distributions. Future work will focus on developing acquisition criteria tailored to distribution-valued outputs. First step in this direction are taken in
\citet{gautier_goal-oriented_2021, gautier_discrete_2026}, where acquisition criteria are built on the SLGP posterior. Criteria that target uncertainty on distributional functionals directly, such as quantiles or spatial integrals, remain to be developed.

Finally, another extension of the present framework is the modeling of multivariate responses, with a possible mix between discrete and continuous output. Incorporating dependence structures and exploring multiple-output SLGP would substantially broaden its applicability.

\section*{Acknowledgements}
This work was largely supported by the Swiss National Science Foundation (SNSF) under Grant No.~\textit{178858}. The author would like to thank David Ginsbourger for his proofreading of earlier versions of the manuscript and for numerous insightful suggestions regarding SLGP's study. The author is also particularly grateful to Yves Deville for his contributions during the early stages of the software development, as well as for methodological suggestions that helped shape the implementation choices presented in this work.

\section*{Use of AI-assisted tools}
Generative AI tools were used in the preparation of this manuscript for typesetting assistance and for grammar and style checking. They were not used to generate the methodology, the software implementation, the numerical experiments or their interpretation. The author takes full responsibility for the content of the paper and of the accompanying package.

\bibliographystyle{plain}
\bibliography{references.bib}  

@Manual{SLGP_R,
    title = {SLGP: Spatial Logistic Gaussian Process for Field Density Estimation},
    author = {Athénaïs Gautier},
    year = {2026},
    note = {R package version 2.0.0},
    url = {https://CRAN.R-project.org/package=SLGP}
}

@misc{sutherland_error_2015,
	title = {On the {Error} of {Random} {Fourier} {Features}},
	url = {http://arxiv.org/abs/1506.02785},
	doi = {10.48550/arXiv.1506.02785},
	urldate = {2024-02-22},
	publisher = {arXiv},
	author = {Sutherland, Danica J. and Schneider, Jeff},
	month = jun,
	year = {2015},
	note = {arXiv:1506.02785 [cs, stat]},
}

@Book{ggplot2,
    author = {Hadley Wickham},
    title = {ggplot2: Elegant Graphics for Data Analysis},
    publisher = {Springer-Verlag New York},
    year = {2016},
    isbn = {978-3-319-24277-4},
    url = {https://ggplot2.tidyverse.org},
}

@misc{stan_development_team_rstan_2024,
	title = {{RStan}: the {R} interface to {Stan}},
	url = {https://mc-stan.org/},
	author = {{Stan Development Team}},
	year = {2024},
}

@article{donner_efficient_2018,
	title = {Efficient {Bayesian} inference for a {Gaussian} process density model},
	journal = {arXiv preprint arXiv:1805.11494},
	author = {Donner, Christian and Opper, Manfred},
	year = {2018},
}

@inproceedings{rahimi_weighted_2009,
	title = {Weighted sums of random kitchen sinks: {Replacing} minimization with randomization in learning},
	booktitle = {Advances in neural information processing systems},
	author = {Rahimi, Ali and Recht, Benjamin},
	year = {2009},
	pages = {1313--1320},
}

@article{lenk_towards_1991,
	title = {Towards a practicable {Bayesian} nonparametric density estimator},
	volume = {78},
	number = {3},
	journal = {Biometrika},
	author = {Lenk, Peter J.},
	year = {1991},
	pages = {531--543},
}

@article{tokdar_towards_2007,
	title = {Towards a faster implementation of density estimation with logistic {Gaussian} process priors},
	volume = {16},
	number = {3},
	journal = {Journal of Computational and Graphical Statistics},
	author = {Tokdar, Surya T.},
	year = {2007},
	note = {Publisher: Taylor \& Francis},
	pages = {633--655},
}

@article{aitchison_statistical_1982,
	title = {The statistical analysis of compositional data},
	volume = {44},
	number = {2},
	journal = {Journal of the Royal Statistical Society: Series B (Methodological)},
	author = {Aitchison, John},
	year = {1982},
	note = {Publisher: Wiley Online Library},
	pages = {139--160},
}

@article{lenk_logistic_1988,
	title = {The logistic normal distribution for {Bayesian}, nonparametric, predictive densities},
	volume = {83},
	number = {402},
	journal = {Journal of the American Statistical Association},
	author = {Lenk, Peter J.},
	year = {1988},
	pages = {509--516},
}

@article{berger_objective_2001,
	title = {Objective {Bayesian} analysis of spatially correlated data},
	volume = {96},
	number = {456},
	journal = {Journal of the American Statistical Association},
	author = {Berger, James O. and De Oliveira, Victor and Sansó, Bruno},
	year = {2001},
	note = {Publisher: Taylor \& Francis},
	pages = {1361--1374},
}

@article{guntuboyina_nonparametric_2018,
	title = {Nonparametric shape-restricted regression},
	volume = {33},
	number = {4},
	journal = {Statistical Science},
	author = {Guntuboyina, Adityanand and Sen, Bodhisattva},
	year = {2018},
	note = {Publisher: JSTOR},
	pages = {568--594},
}

@article{papamakarios_neural_2019,
	title = {Neural density estimation and likelihood-free inference},
	journal = {arXiv preprint arXiv:1910.13233},
	author = {Papamakarios, George},
	year = {2019},
}

@article{hall_methods_1999,
	title = {Methods for estimating a conditional distribution function},
	volume = {94},
	number = {445},
	journal = {Journal of the American Statistical association},
	author = {Hall, Peter and Wolff, Rodney C. L. and Yao, Qiwei},
	year = {1999},
	note = {Publisher: Taylor \& Francis Group},
	pages = {154--163},
}

@article{walker_sampling_2007,
	title = {Sampling the {Dirichlet} mixture model with slices},
	volume = {36},
	number = {1},
	journal = {Communications in Statistics—Simulation and Computation®},
	author = {Walker, Stephen G.},
	year = {2007},
	note = {Publisher: Taylor \& Francis},
	pages = {45--54},
}

@article{papaspiliopoulos_retrospective_2008,
	title = {Retrospective {Markov} chain {Monte} {Carlo} methods for {Dirichlet} process hierarchical models},
	volume = {95},
	number = {1},
	journal = {Biometrika},
	author = {Papaspiliopoulos, Omiros and Roberts, Gareth O.},
	year = {2008},
	note = {Publisher: Oxford University Press},
	pages = {169--186},
}

@inproceedings{rahimi_random_2008,
	title = {Random features for large-scale kernel machines},
	booktitle = {Advances in neural information processing systems},
	author = {Rahimi, Ali and Recht, Benjamin},
	year = {2008},
	pages = {1177--1184},
}

@article{tokdar_posterior_2007,
	title = {Posterior consistency of logistic {Gaussian} process priors in density estimation},
	volume = {137},
	number = {1},
	journal = {Journal of statistical planning and inference},
	author = {Tokdar, Surya T. and Ghosh, Jayanta K.},
	year = {2007},
	note = {Publisher: Elsevier},
	pages = {34--42},
}

@article{papamakarios_masked_2017,
	title = {Masked autoregressive flow for density estimation},
	volume = {30},
	journal = {Advances in neural information processing systems},
	author = {Papamakarios, George and Pavlakou, Theo and Murray, Iain},
	year = {2017},
}

@inproceedings{dwivedi_log-concave_2018,
	title = {Log-concave sampling: {Metropolis}-{Hastings} algorithms are fast!},
	booktitle = {Conference on learning theory},
	publisher = {PMLR},
	author = {Dwivedi, Raaz and Chen, Yuansi and Wainwright, Martin J. and Yu, Bin},
	year = {2018},
	pages = {793--797},
}

@article{khinchin_korrelationstheorie_1934,
	title = {Korrelationstheorie der stationären stochastischen {Prozesse}},
	volume = {109},
	number = {1},
	journal = {Mathematische Annalen},
	author = {Khinchin, Alexander},
	year = {1934},
	note = {Publisher: Springer},
	pages = {604--615},
}

@article{egozcue_hilbert_2006,
	title = {Hilbert space of probability density functions based on {Aitchison} geometry},
	volume = {22},
	number = {4},
	journal = {Acta Mathematica Sinica},
	author = {Egozcue, Juan José and Díaz–Barrero, José Luis and Pawlowsky–Glahn, Vera},
	year = {2006},
	note = {Publisher: Springer},
	pages = {1175--1182},
}

@article{gautier_goal-oriented_2021,
	title = {Goal-oriented adaptive sampling under random field modelling of response probability distributions},
	volume = {71},
	copyright = {All rights reserved},
	url = {https://doi.org/10.1051/proc/202171108},
	doi = {10.1051/proc/202171108},
	journal = {ESAIM: Proceedings and Surveys},
	author = {Gautier, Athénaïs and Ginsbourger, David and Pirot, Guillaume},
	year = {2021},
	pages = {89--100},
}

@misc{gautier_athenais_github_imple,
	title = {Github repository for: Accompanying code of: Spatial Logistic Gaussian Process {SLGP}: Implementing a nonparametric approach to conditional density estimation on heterogeneous data },
	url = {https://github.com/AthenaisGautier/SLGPImplementation},
	publisher = {GitHub},
	author = {{Gautier, Athénaïs}},
	year = {2025},
	note = {},
}

@article{fan_estimation_1996,
	title = {Estimation of conditional densities and sensitivity measures in nonlinear dynamical systems},
	volume = {83},
	number = {1},
	journal = {Biometrika},
	author = {Fan, Jianqing and Yao, Qiwei and Tong, Howell},
	year = {1996},
	note = {Publisher: Oxford University Press},
	pages = {189--206},
}

@article{zhu_emulation_2020,
	title = {Emulation of stochastic simulators using generalized lambda models},
	journal = {arXiv preprint arXiv:2007.00996},
	author = {Zhu, Xujia and Sudret, Bruno},
	year = {2020},
}

@article{leonard_density_1978,
	title = {Density estimation, stochastic processes and prior information},
	volume = {40},
	number = {2},
	journal = {Journal of the Royal Statistical Society, Series B: Methodological},
	author = {Leonard, Tom},
	year = {1978},
	pages = {113--132},
}

@article{gautier_continuous_2021,
	title = {Continuous logistic {Gaussian} random measure fields for spatial distributional modelling},
	issn = {1572-9052},
	url = {https://doi.org/10.1007/s10463-025-00968-3},
	doi = {10.1007/s10463-025-00968-3},
	journal = {Annals of the Institute of Statistical Mathematics},
	author = {Gautier, Athénaïs and Ginsbourger, David},
	month = sep,
	year = {2026},
}

@article{fuglstad_constructing_2019,
	title = {Constructing priors that penalize the complexity of {Gaussian} random fields},
	volume = {114},
	number = {525},
	journal = {Journal of the American Statistical Association},
	author = {Fuglstad, Geir-Arne and Simpson, Daniel and Lindgren, Finn and Rue, Håvard},
	year = {2019},
	note = {Publisher: Taylor \& Francis},
	pages = {445--452},
}

@article{rothfuss_conditional_2019,
	title = {Conditional density estimation with neural networks: {Best} practices and benchmarks},
	journal = {arXiv preprint arXiv:1903.00954},
	author = {Rothfuss, Jonas and Ferreira, Fabio and Walther, Simon and Ulrich, Maxim},
	year = {2019},
}

@misc{rojas_conditional_2005,
	title = {Conditional density estimation using finite mixture models with an application to astrophysics},
	publisher = {Carnegie Mellon University, Pittsburgh, USA.[Online].},
	author = {Rojas, Alex L. and Genovese, Christopher R. and Miller, Christopher J. and Nichol, Robert and Wasserman, Larry},
	year = {2005},
}

@article{talska_compositional_2018,
	title = {Compositional regression with functional response},
	volume = {123},
	journal = {Computational Statistics \& Data Analysis},
	author = {Talská, R. and Menafoglio, Alessandra and Machalová, Jitka and Hron, Karel and Fišerová, E.},
	year = {2018},
	note = {Publisher: Elsevier},
	pages = {66--85},
}

@article{mclachlan_finite_2004,
  title={Finite mixture models},
  author={McLachlan, Geoffrey J and Lee, Sharon X and Rathnayake, Suren I},
  journal={Annual review of statistics and its application},
  volume={6},
  number={1},
  pages={355--378},
  year={2019},
  publisher={Annual Reviews}
}

@article{hoffman_NUTS_2014,
  title={The No-U-Turn sampler: adaptively setting path lengths in Hamiltonian Monte Carlo.},
  author={Hoffman, Matthew D and Gelman, Andrew and others},
  journal={J. Mach. Learn. Res.},
  volume={15},
  number={1},
  pages={1593--1623},
  year={2014}
}

@book{silverman_density_1986,
  title={Density estimation for statistics and data analysis},
  author={Silverman, Bernard W},
  year={2018},
  publisher={Routledge}
}

@article{ferguson_bayesian_1973,
  title={A Bayesian analysis of some nonparametric problems},
  author={Ferguson, Thomas S},
  journal={The annals of statistics},
  pages={209--230},
  year={1973},
  publisher={JSTOR}
}

@article{escobar_bayesian_1995,
  title={Bayesian density estimation and inference using mixtures},
  author={Escobar, Michael D and West, Mike},
  journal={Journal of the american statistical association},
  volume={90},
  number={430},
  pages={577--588},
  year={1995},
  publisher={Taylor \& Francis}
}

@article{tokdar_bayesian_2010,
	title = {Bayesian density regression with logistic {Gaussian} process and subspace projection},
	volume = {5},
	number = {2},
	journal = {Bayesian analysis},
	author = {Tokdar, Surya T. and Zhu, Yu M. and Ghosh, Jayanta K.},
	year = {2010},
	note = {Publisher: International Society for Bayesian Analysis},
	pages = {319--344},
}

@article{van_der_vaart_adaptive_2009,
	title = {Adaptive {Bayesian} estimation using a {Gaussian} random field with inverse gamma bandwidth},
	volume = {37},
	number = {5B},
	journal = {The Annals of Statistics},
	author = {van der Vaart, Aad W. and van Zanten, J. Harry},
	year = {2009},
	note = {Publisher: Institute of Mathematical Statistics},
	pages = {2655--2675},
}

@article{menafoglio_universal_2013,
	title = {A {Universal} {Kriging} predictor for spatially dependent functional data of a {Hilbert} {Space}},
	volume = {7},
	journal = {Electronic Journal of Statistics},
	author = {Menafoglio, Alessandra and Secchi, Piercesare and Dalla Rosa, Matilde},
	year = {2013},
	note = {Publisher: The Institute of Mathematical Statistics and the Bernoulli Society},
	pages = {2209--2240},
}

@article{jain_split-merge_2004,
	title = {A split-merge {Markov} chain {Monte} {Carlo} procedure for the {Dirichlet} process mixture model},
	volume = {13},
	number = {1},
	journal = {Journal of computational and Graphical Statistics},
	author = {Jain, Sonia and Neal, Radford M.},
	year = {2004},
	note = {Publisher: Taylor \& Francis},
	pages = {158--182},
}

@article{menafoglio_kriging_2014,
	title = {A kriging approach based on {Aitchison} geometry for the characterization of particle-size curves in heterogeneous aquifers},
	volume = {28},
	number = {7},
	journal = {Stochastic Environmental Research and Risk Assessment},
	author = {Menafoglio, Alessandra and Guadagnini, Alberto and Secchi, Piercesare},
	year = {2014},
	note = {Publisher: Springer},
	pages = {1835--1851},
}

@phdthesis{gautier_modelling_2023,
	type = {{PhD} {Thesis}},
	title = {Modelling and predicting distribution-valued fields with applications to inversion under uncertainty},
	copyright = {All rights reserved},
	school = {Institute of Mathematical Statistics and Actuarial Science, University of Bern},
	author = {Gautier, Athénaïs},
	year = {2023},
}

@misc{muzellec_learning_2022,
	title = {Learning {PSD}-valued functions using kernel sums-of-squares},
	url = {http://arxiv.org/abs/2111.11306},
	doi = {10.48550/arXiv.2111.11306},
	urldate = {2024-01-22},
	publisher = {arXiv},
	author = {Muzellec, Boris and Bach, Francis and Rudi, Alessandro},
	month = jan,
	year = {2022},
	note = {arXiv:2111.11306 [cs, stat]},
}

@inproceedings{murray_gaussian_2008,
	title = {The {Gaussian} {Process} {Density} {Sampler}},
	volume = {21},
	url = {https://papers.nips.cc/paper_files/paper/2008/hash/335f5352088d7d9bf74191e006d8e24c-Abstract.html},
	urldate = {2024-01-22},
	booktitle = {Advances in {Neural} {Information} {Processing} {Systems}},
	publisher = {Curran Associates, Inc.},
	author = {Murray, Iain and MacKay, David and Adams, Ryan P},
	year = {2008},
}

@inproceedings{rudi_psd_2021,
	title = {{PSD} {Representations} for {Effective} {Probability} {Models}},
	volume = {34},
	url = {https://proceedings.neurips.cc/paper/2021/hash/a1b63b36ba67b15d2f47da55cdb8018d-Abstract.html},
	urldate = {2024-01-22},
	booktitle = {Advances in {Neural} {Information} {Processing} {Systems}},
	publisher = {Curran Associates, Inc.},
	author = {Rudi, Alessandro and Ciliberto, Carlo},
	year = {2021},
	pages = {19411--19422},
}

@Manual{quantreg,
    title = {quantreg: Quantile Regression},
    author = {Roger Koenker},
    year = {2025},
    note = {R package version 6.1},
    url = {https://CRAN.R-project.org/package=quantreg},
}

@book{koenker2005quantile,
  title={Quantile regression},
  author={Koenker, Roger},
  volume={38},
  year={2005},
  publisher={Cambridge university press}
}

@phdthesis{zhu_surrogate_2023,
	title = {Surrogate modeling for stochastic simulators using statistical approaches},
	school = {ETH Zurich},
	author = {Zhu, Xujia},
	year = {2023},
}

@Manual{r_core_team,
  title        = {{R}: A Language and Environment for Statistical Computing},
  author       = {{R Core Team}},
  organization = {R Foundation for Statistical Computing},
  address      = {Vienna, Austria},
  year         = {2025},
  url          = {https://www.R-project.org/},
}

@Article{isacks_seismology_1968,
  title   = {Seismology and the New Global Tectonics},
  author  = {Isacks, Bryan and Oliver, Jack and Sykes, Lynn R.},
  journal = {Journal of Geophysical Research},
  year    = {1968},
  volume  = {73},
  number  = {18},
  pages   = {5855--5899},
  doi     = {10.1029/JB073i018p05855},
}

@Book{frohlich_deep_2006,
  title     = {Deep Earthquakes},
  author    = {Frohlich, Cliff},
  publisher = {Cambridge University Press},
  address   = {Cambridge},
  year      = {2006},
}

@Article{hayfield_np_2008,
  title   = {Nonparametric Econometrics: The {np} Package},
  author  = {Hayfield, Tristen and Racine, Jeffrey S.},
  journal = {Journal of Statistical Software},
  year    = {2008},
  volume  = {27},
  number  = {5},
  pages   = {1--32},
  doi     = {10.18637/jss.v027.i05},
}

@Article{hyndman_estimating_1996,
  title   = {Estimating and Visualizing Conditional Densities},
  author  = {Hyndman, Rob J. and Bashtannyk, David M. and Grunwald, Gary K.},
  journal = {Journal of Computational and Graphical Statistics},
  year    = {1996},
  volume  = {5},
  number  = {4},
  pages   = {315--336},
  doi     = {10.1080/10618600.1996.10474715},
}

@Article{rigby_gamlss_2005,
  title   = {Generalized Additive Models for Location, Scale and Shape},
  author  = {Rigby, Robert A. and Stasinopoulos, D. Mikis},
  journal = {Journal of the Royal Statistical Society C (Applied Statistics)},
  year    = {2005},
  volume  = {54},
  number  = {3},
  pages   = {507--554},
  doi     = {10.1111/j.1467-9876.2005.00510.x},
}

@Article{umlauf_bamlss_2021,
  title   = {{bamlss}: A {L}ego Toolbox for Flexible {B}ayesian Regression (and Beyond)},
  author  = {Umlauf, Nikolaus and Klein, Nadja and Simon, Thorsten and Zeileis, Achim},
  journal = {Journal of Statistical Software},
  year    = {2021},
  volume  = {100},
  number  = {4},
  pages   = {1--53},
  doi     = {10.18637/jss.v100.i04},
}

@Article{gautier_discrete_2026,
  author  = {Ath\'ena\"is Gautier},
  title   = {Spatial Logistic {G}aussian Processes Based Surrogates for
             Spatially Dependent Discrete Outputs: Modelling, Uncertainty
             Quantification, and Data Acquisition},
  journal = {Statistical Papers},
  year    = {2026},
  note    = {Accepted for the mODa 14 special issue},
}

\newpage
\appendix
\setcounter{figure}{0}
\renewcommand{\thefigure}{A-\arabic{figure}}
\section{Appendix: Proof of the log-concavity result}
\label{Appendix:convexity}
\begin{proof}[Proof of Theorem~\ref{th_log-concavity}]
To prove this statement, we compute the gradient and Hessian of the negative log likelihood. The negative log-likelihood is:
    
	\begin{align}
		\label{eq:negloglikelihood}
		\ell( \bs \epsilon \vert \mathbf{t} ) &= 
		 \sum_{j=1}^n \left( - \sum_{i=1}^p f_i(\xX_j, t_j)  \epsilon_i + \log \left( \displaystyle \int_\xT e^{\sum_{i=1}^p f_{i}(\xX_j, u) \epsilon_i}\,du \right)  \right)
	\end{align}
	
	Its gradients write:
	
	\begin{align}
		\label{eq:gradnegloglikelihood}
		\dfrac{\partial \ell( \bs \epsilon \vert \mathbf{t} )}{\partial \epsilon_{i}} &= 
		- \sum_{j=1}^n f_i(\xX_j, t_j) +
		\sum\limits_{j=1}^n  \displaystyle \int_\xT f_{i}(\xX_j, u) \dfrac{ e^{\bs \epsilon^T  \bs F({\xX}_j, u) } }{\int_\xT e^{\bs \epsilon^T  \bs F({\xX}_j, v)} \,dv} \,du 
	\end{align}
	The second order derivatives are:
	
	\begin{equation}
		\label{eq:hessnegloglikelihood}
		\begin{array}{c}
			\dfrac{\partial^2 \ell( \bs \epsilon \vert \mathbf{t} )}{\partial \epsilon_{i}\partial \epsilon_{i'}} = 
			\sum\limits_{j=1}^n  \displaystyle \int_\xT f_{i}(\xX_j, u)  f_{i'}(\xX_j, u)  \dfrac{ e^{\bs \epsilon^T F(\xX_j, u) } }{\int_\xT e^{\bs \epsilon^T F(\xX_j, v)} \,dv} \,du \\
			-  \sum\limits_{j=1}^n  \left(\displaystyle \int_\xT f_{i}(\xX_j, u) \dfrac{ e^{\bs \epsilon^T F(\xX_j, u) } }{\int_\xT e^{\bs \epsilon^T F(\xX_j, v)} \,dv} \,du \right)  \left( \displaystyle \int_\xT f_{i'}(\xX_j, u) \dfrac{ e^{\bs \epsilon^T F(\xX_j, u) } }{\int_\xT e^{\bs \epsilon^T F(\xX_j, v)} \,dv} \,du \right) 
		\end{array}
	\end{equation}

    First, observe that only the first term in Equation~\ref{eq:gradnegloglikelihood} depends on the values of the $t_i$s, and that all the other terms can be computed only once per distinct value of the $\xX_i$s. Let us assume there are $K \leq n$ unique values of the $\xX_i$s. We denote $\tilde \xX_1, ..., \tilde \xX_K$ these values, and assume that $\tilde \xX_k$ appears $n_k$ times in $ \{\xX_i\}_{1 \leq i \leq n}$.
    Then: 

	\begin{align}
		\label{eq:gradnegloglikelihood2}
		\dfrac{\partial \ell( \bs \epsilon \vert \mathbf{t} )}{\partial \epsilon_{i}} &= 
		- \sum_{j=1}^n f_i(\xX_j, t_j) +
		\sum\limits_{k=1}^K n_k \displaystyle \int_\xT f_{i}(\tilde \xX_k, u) \dfrac{ e^{\bs \epsilon^T  F(\tilde \xX_k, u) } }{\int_\xT e^{\bs \epsilon^T  F(\tilde \xX_k, v)} \,dv} \,du 
	\end{align}
    and
    \begin{equation}
		\label{eq:hessnegloglikelihood2}
		\begin{array}{c}
			\dfrac{\partial^2 \ell( \bs \epsilon \vert \mathbf{t} )}{\partial \epsilon_{i}\partial \epsilon_{i'}} = 
			\sum\limits_{k=1}^K n_k  \displaystyle \int_\xT f_{i}(\tilde \xX_k, u)  f_{i'}(\tilde \xX_k, u)  \dfrac{ e^{\bs \epsilon^T F(\tilde \xX_k, u) } }{\int_\xT e^{\bs \epsilon^T F(\tilde \xX_k, v)} \,dv} \,du \\
			-  \sum\limits_{k=1}^K n_k  \left(\displaystyle \int_\xT f_{i}(\tilde \xX_k, u) \dfrac{ e^{\bs \epsilon^T F(\tilde \xX_k, u) } }{\int_\xT e^{\bs \epsilon^T F(\tilde \xX_k, v)} \,dv} \,du \right)  \left( \displaystyle \int_\xT f_{i'}(\tilde \xX_k, u) \dfrac{ e^{\bs \epsilon^T F(\tilde \xX_k, u) } }{\int_\xT e^{\bs \epsilon^T F(\tilde \xX_k, v)} \,dv} \,du \right) 
		\end{array}
	\end{equation}

	For a given $\bs \epsilon$, let us introduce $Y_1, ... Y_K$, $K$ independent random variables where each of the $Y_k$ has a probability density of $\dfrac{ e^{\bs \epsilon^T F(\tilde \xX_k, u) } }{\int_\xT e^{\bs \epsilon^T F(\tilde \xX_k, v)} \,dv} $. Then:
	\begin{equation}
		\begin{split}
			\label{eq:hessnegloglikelihoodfollowup}	\dfrac{\partial^2 \ell( \bs \epsilon \vert \mathbf{t} )}{\partial \epsilon_{i}\partial \epsilon_{i'}} & = 
			\sum\limits_{k=1}^K n_k  \left( \mathbb{E} \left[ f_{i}(\tilde \xX_k, Y_k)  f_{i'}(\tilde \xX_k, Y_{k})  \right] -  \mathbb{E} \left[ f_{i}(\tilde \xX_k, Y_k)  \right]  \mathbb{E} \left[ f_{i'}(\tilde \xX_k, Y_{k})  \right]  \right) \\
			& = \sum\limits_{k=1}^K n_k \text{Cov} \left( f_{i}(\tilde \xX_k, Y_k), f_{i'}(\tilde \xX_k, Y_{k}) \right)
		\end{split}
	\end{equation}
	Since the Hessian matrix is a (sum of) covariance matrices, it inherits their symmetry and positive-semidefiniteness. This suffices to prove that the negative log-likelihood is convex, and therefore that the likelihood is log-concave.

    Moreover, the Gaussian prior contributes $\frac12\Vert \bs \epsilon\Vert^2$, which makes the posterior strictly log-concave and ensures uniqueness in an optimisation setting (MAP).
	\end{proof}

\clearpage
\section{On the normalising integral of the SLGP}

\label{Appendix:approxIntegral}
We recall that our goal is to provide a reasonably fast and numerically stable approximation of the negative log likelihood in a SLGP model.

\begin{align}
    \ell(\boldsymbol{\epsilon} \vert  D) &= -  \sigma \boldsymbol{\epsilon}^{\top} \sum^{n}_{i=1} \bs F(\xX_i, t_i) +  \sum^{n}_{i=1} \left(\log \int_0^1 \exp\{\sigma\boldsymbol{\epsilon}^{\top} \bs F(\xX_i, u)\} \,du \right)
\end{align}

Throughout this appendix $n$ denotes the number of observations and $m$ the number of quadrature nodes $(u_j)_{j=1}^m$ of a regular grid of $\xT = [0, 1]$. First, we note that when setting $M_i:= \max_{u} \sigma\boldsymbol{\epsilon}^{\top} \bs F(\xX_i, u)$, we have the equality:

\begin{align}
    \ell(\boldsymbol{\epsilon} \vert  D) &= - \sigma \boldsymbol{\epsilon}^{\top} \sum^{n}_{i=1} \bs F(\xX_i, t_i) +  \sum_{i=1}^n M_i + \sum^{n}_{i=1} \left(\log \int_0^1 \exp\{\sigma\boldsymbol{\epsilon}^{\top} \bs F(\xX_i, u) - M_i\} \,du \right)
\end{align}

This simple numerical manipulations allows for greater numerical stability, ensuring that the exponentials will not go out of range, and that the Riemann's sum approximation of the integrals can be performed without issues.

\begin{align}
    \ell(\boldsymbol{\epsilon} \vert  D) & \approx - \sigma \boldsymbol{\epsilon}^{\top} \sum^{n}_{i=1} \bs F(\xX_i, t_i) +  \sum_{i=1}^n M_i + \sum^{n}_{i=1} \log \left( \frac{1}{m} \sum^{m}_{j=1} \exp\{\sigma \boldsymbol{\epsilon}^{\top} \bs F(\xX_i, u_j) -M_i \} \right)
\end{align}   

While numerically stable, this term remains fairly expensive to evaluate, as it involves approximating one integral (and the maximum values $M_i$) per distinct value of the predictor $\xX$. However, when the family of basis functions consist of smooth elements, for any value of $\bs \epsilon$, the term $\int_\xT e^{\bs \epsilon^T  F({\xX}, u)} \,du $ will evolve smoothly with $\xX$, and as such, one can enjoy further approximations in the model to greatly alleviate numerical costs. We already mentioned in the main body of the paper that we propose three main approaches, we further detail them here, complete with equations and illustrations.

\paragraph{Approximation by nearest neighbour on a regular grid}
To avoid computing up to $n$ integrals, we can introduce a grid of $D \times \xT$ and denote $\tilde{\xX}_i$ the closest neighbour of $\xX_i$ on this grid and $\tilde{u}_i$ that of $t_i$.

\begin{align}
\ell(\boldsymbol{\epsilon} \vert  D) \approx \ell_{NN}(\boldsymbol{\epsilon} \vert  D)  \approx  - \sigma \boldsymbol{\epsilon}^{\top} \sum^{n}_{i=1} F(\tilde{x}_i,  \tilde{u}_i)+  \sum^{n}_{i=1}\log \left( \frac{1}{m} \sum^{m}_{j=1} \exp\{\sigma \boldsymbol{\epsilon}^{\top}F(\tilde{x}_i, u_j) \} \right)
\end{align}

This enables one to bound the maximum number of integrals to compute. However, it yields a piecewise constant approximation of the negative log-likelihood landscape, which can be too rough.

\paragraph{Weighted nearest neighbours} Yet another suitable approximation is to use a weighted combination of several neighbours.

For a point $[\xX_i, t_i]$, we denote $\left[ \check{\xX}_{ij}, \check{t}_{ij} \right]_{j=1}^{2^{\dimX+1}}$ the $2^{\dimX+1}$ vertices of the smallest enclosing hypercube within a regular grid. Considering weight factors $w_{ij}$ (with $\sum_j w_{ij} = 1$), we get:

\begin{align}
\ell(\boldsymbol{\epsilon} \vert  D) \approx  \ell_{WNN}(\bs{\epsilon} \vert D) \approx - \sum^{n}_{i=1} \log \left( \sum^{2^{\dimX+1}}_{j=1} w_{ij}  \frac{\exp\{\sigma\bs{\epsilon}^{\top} \bs F(\check{\xX}_{ij},  \check{t}_{ij})\}}{ \frac{1}{m} \sum^{m}_{l=1} \exp\{\sigma \bs{\epsilon}^{\top} \bs F(\check{\xX}_{ij}, u_l) \} }\right)
\end{align}
\clearpage
\begin{figure}[H]
\centering
    \includegraphics[width=0.95\linewidth]{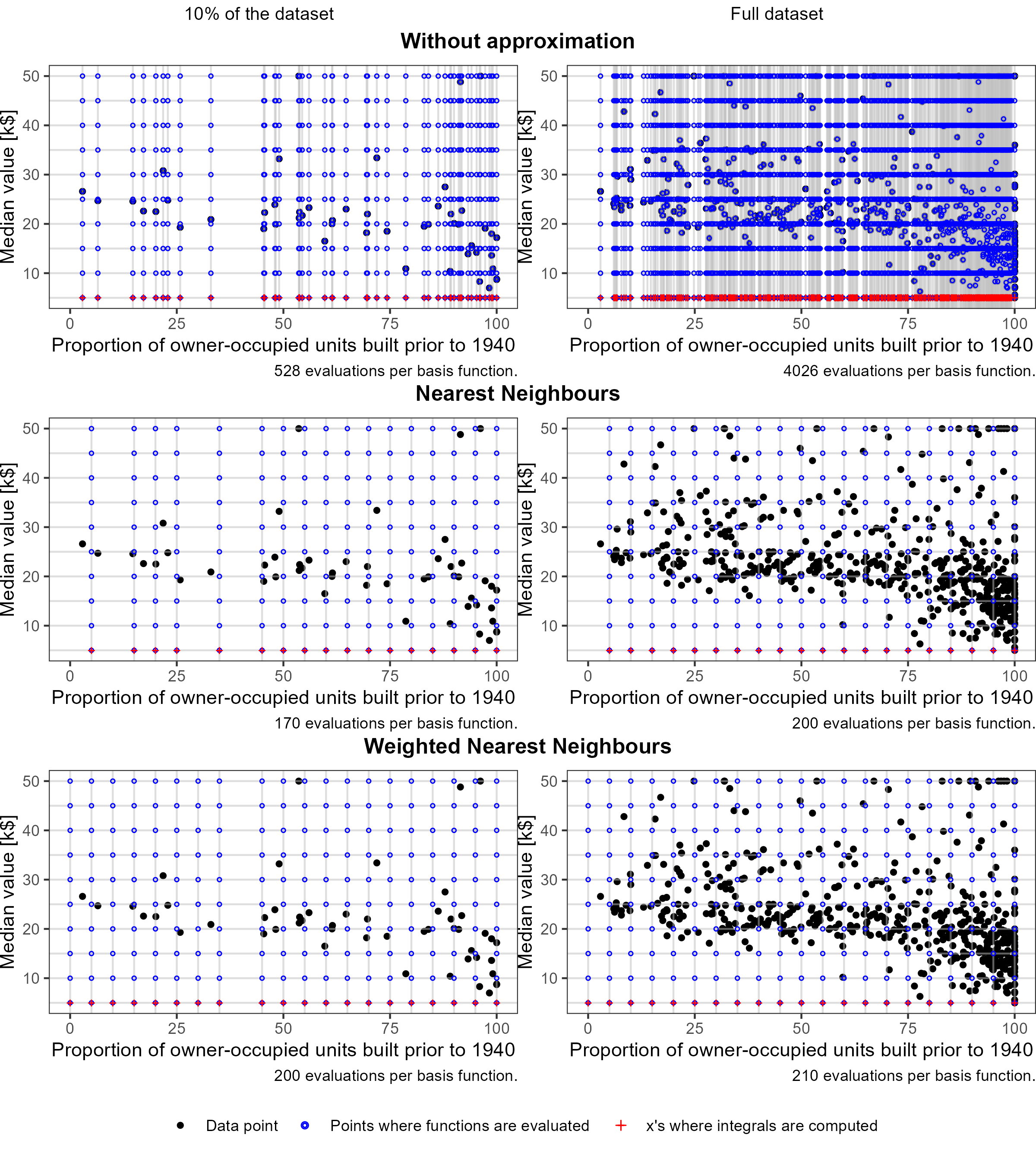}
\caption{Visualisation of three strategies for integral approximation applied to the running example, distinguished by sample sizes. The top panels approximates separately the integral for each predictor variable $\xX$. In the middle panels, the `Nearest Neighbour' approach assigns the value of the closest grid node to each data point. The bottom panels demonstrate the `Weighted Nearest Neighbour' method, calculating the integral value at a data point as a weighted sum of grid node values within the smallest enclosing hypercube}
\label{fig:main}
\end{figure}

\section{Hyper-parameters: variance heuristic and length-scale selection}
\label{app:lengthscale_optim}

\subsection{The variance heuristic}

Figure~\ref{fig:Heuristic_var} shows prior realisations of SLGPs for four values of $\sigma$, expressed through the resulting range $\max_{t} Z_{\xX, t} - \min_{t} Z_{\xX, t}$ of the underlying GP, and for three degrees of smoothness. Small ranges give fields that are almost uniform and carry little structure; large ones give sharply peaked fields that are numerically fragile. A range of about 5 is the compromise retained by the heuristic of Section~\ref{subsec:hyperparam}.
\vspace{-12pt}
\begin{figure}[H]
    \centering
    \includegraphics[width=0.95\linewidth]{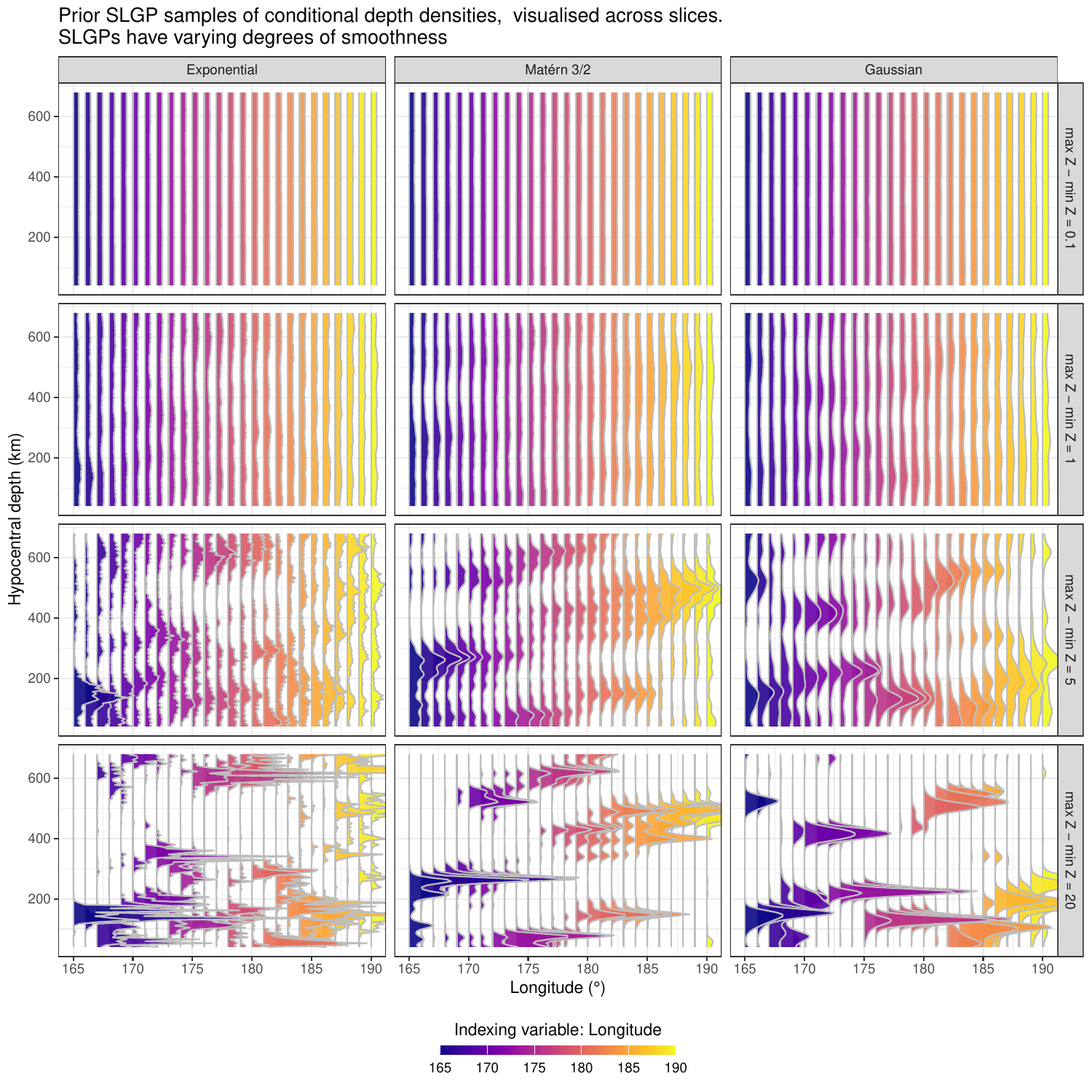}
\caption{Four realisations of SLGPs, with varying ranges of values of the underlying GP.}
\label{fig:Heuristic_var}
\end{figure}

\subsection{Length-scale selection}

To assess the impact of length-scale selection on SLGP estimation, we follow the MAP-based strategy described in Section~\ref{subsec:hyperparam}. We use an Inverse-Gamma prior on each length-scale parameter, perform posterior evaluations and compare optimization results. The model is trained on 75\% of the data and evaluated on the remaining 25\%.

Figure~\ref{fig:optimlandscape} displays the resulting negative log-posterior landscapes.
\begin{figure}[H]
\centering 
    \includegraphics[width=0.7\linewidth]{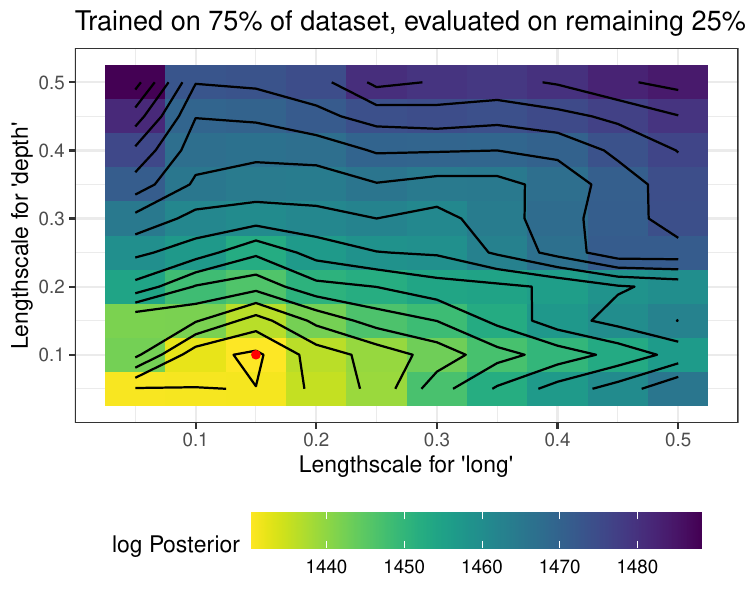}
\caption{Negative log posterior density (up to a constant) over the length-scale grid (\longi, \depth), for a model trained on 75\% of the data and evaluated on the remaining 25\%. The red dot indicates the optimal length-scales.}
\label{fig:optimlandscape}
\end{figure}

This landscape is fairly consistent with our proposed heuristic of setting all length-scales to 15\% of the ranges.

Now, for a qualitative interpretation, we fit and compare in Figures~\ref{fig:plotRibbonsLen1} and~\ref{fig:plotRibbonsLen2} SLGP models with four representative choices of length-scales:
\begin{itemize}
    \item The heuristic length-scale: 15\% of the \longi range and 15\% of the \depth range.
    \item The long length-scale: 50\% of the \longi range and 50\% of the \depth range.
    \item The optimised length-scale: 15\% of the \longi range and 10\% of the \depth range.
    \item The short length-scale: 1\% of the \longi range and 1\% of the \depth range.
\end{itemize}

\begin{figure}[H]
\centering 
    \includegraphics[width=\linewidth]{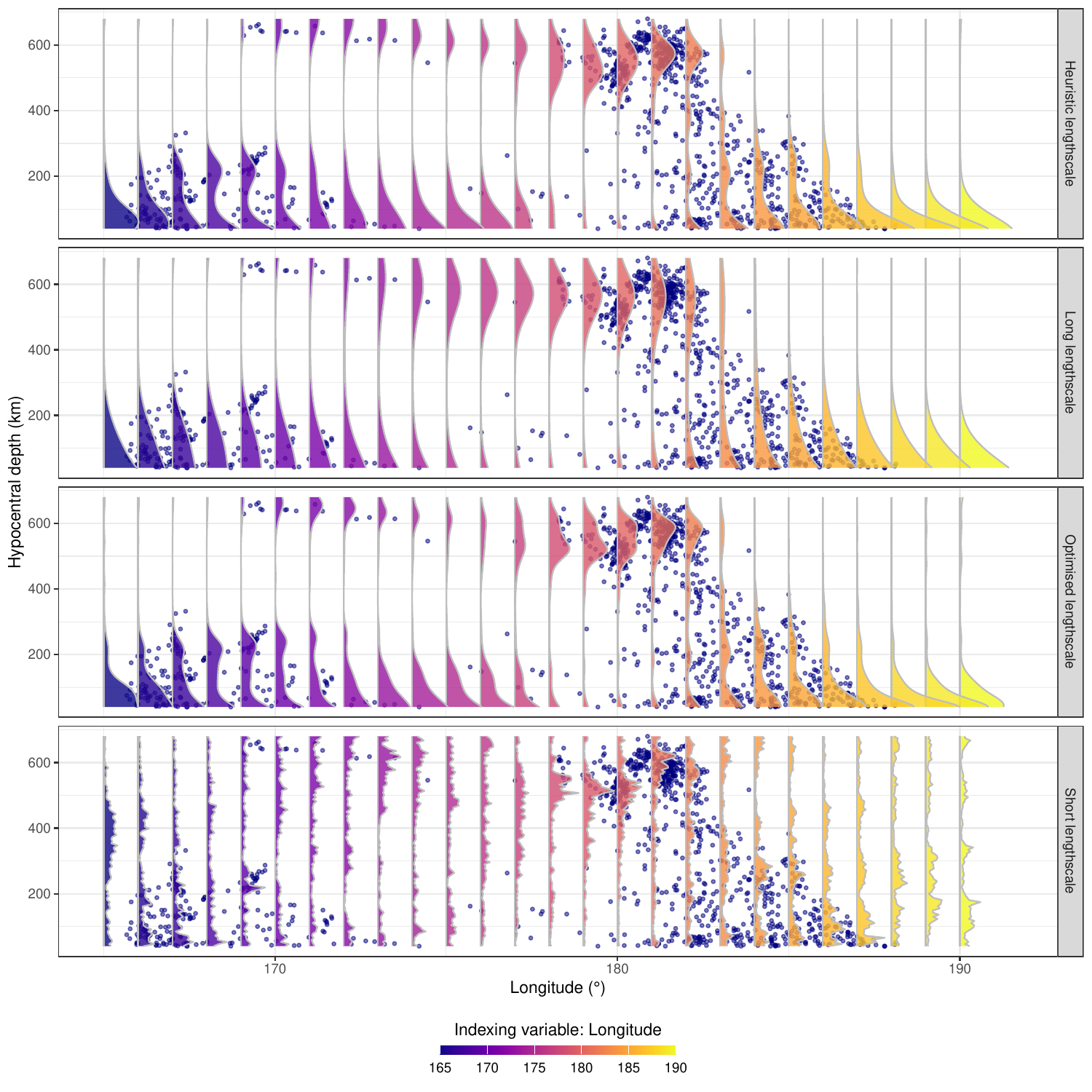}
\caption{Fitted conditional depth densities across \longi under the four length-scale settings.}
\label{fig:plotRibbonsLen1}
\end{figure}

\begin{figure}[H]
\centering 
    \includegraphics[width=\linewidth]{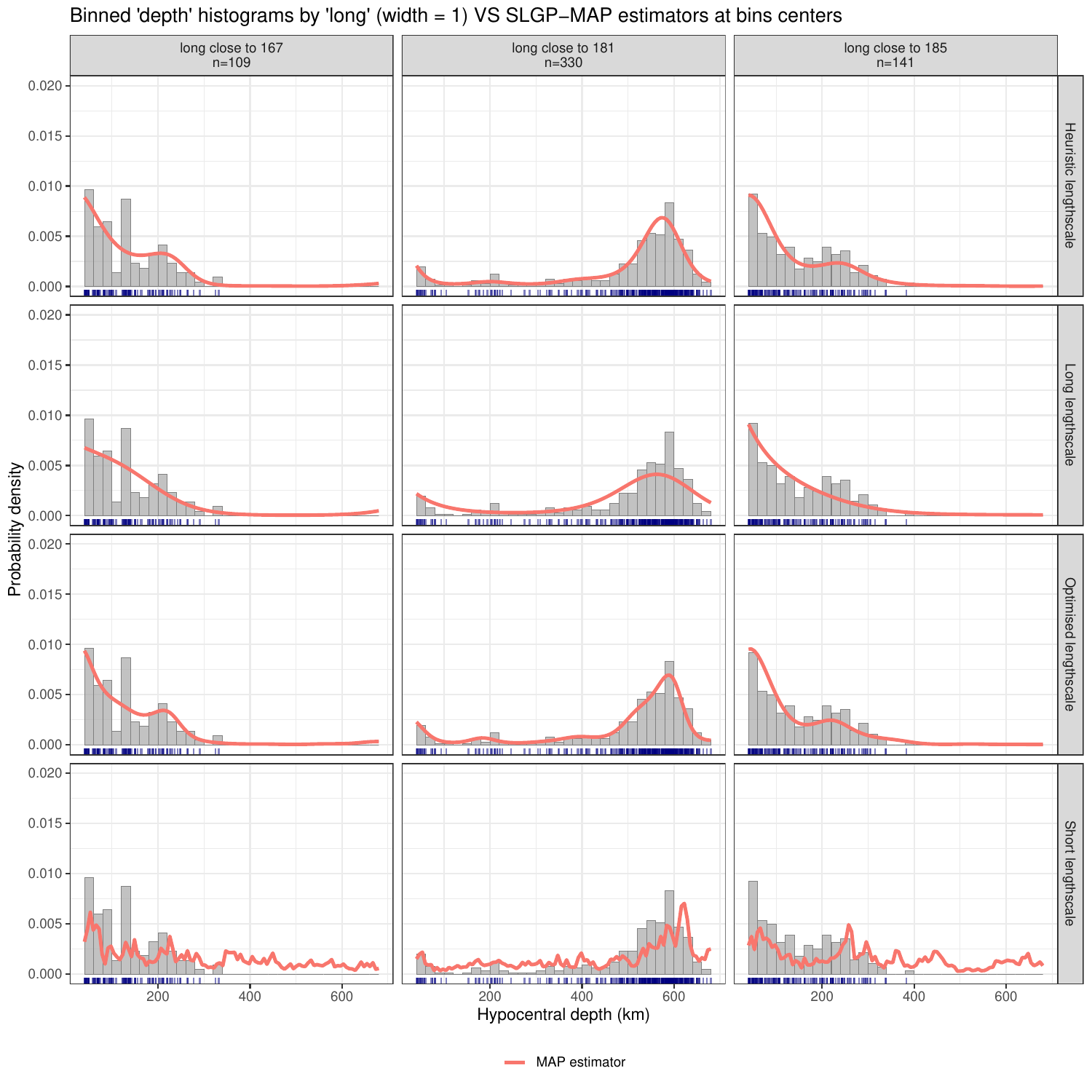}
\caption{Fitted conditional depth densities under the four length-scale settings, against empirical histograms at three \longi slices.}
\label{fig:plotRibbonsLen2}
\end{figure}
Both the heuristic and the optimised length-scales yield a compromise between these two failure modes. For simplicity, we continue with the length-scales given by the heuristic, set to 15\% of the ranges.
\newpage
\section{Local estimation errors under the different sampling schemes}
\label{app:figures}

The figures display the Kullback-Leibler divergence, the Hellinger distance, the Cramér-von Mises criterion and the total variation distance between the estimated and the reference conditional densities as a function of \longi. The qualitative reading is the same in all four cases: the error concentrates where the index has been sparsely or never sampled, and the grid designs leave a visible ripple between their design points.

\begin{figure}[H]
\centering
    \includegraphics[width=0.92\linewidth]{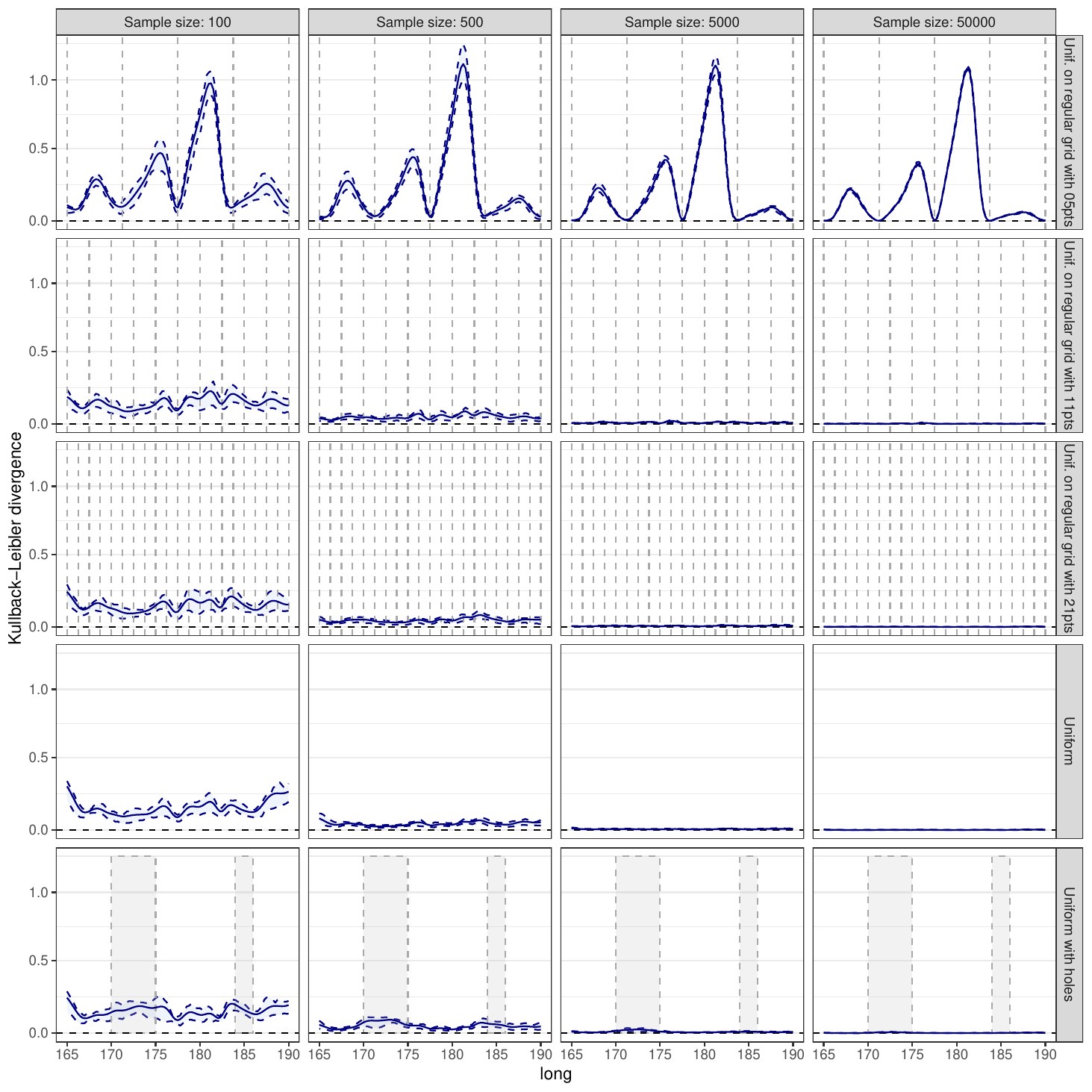}
\caption{Impact of the \longi selection scheme on SLGP density estimation using MAP. Local Kullback–Leibler divergences between the estimated densities and the reference field are shown across varying training sample sizes for different schemes: mean (solid lines) and 25th–75th percentile range (shaded ribbon).}
\label{fig:benchmark3a}
\end{figure}

\begin{figure}[H]
\centering
    \includegraphics[width=0.92\linewidth]{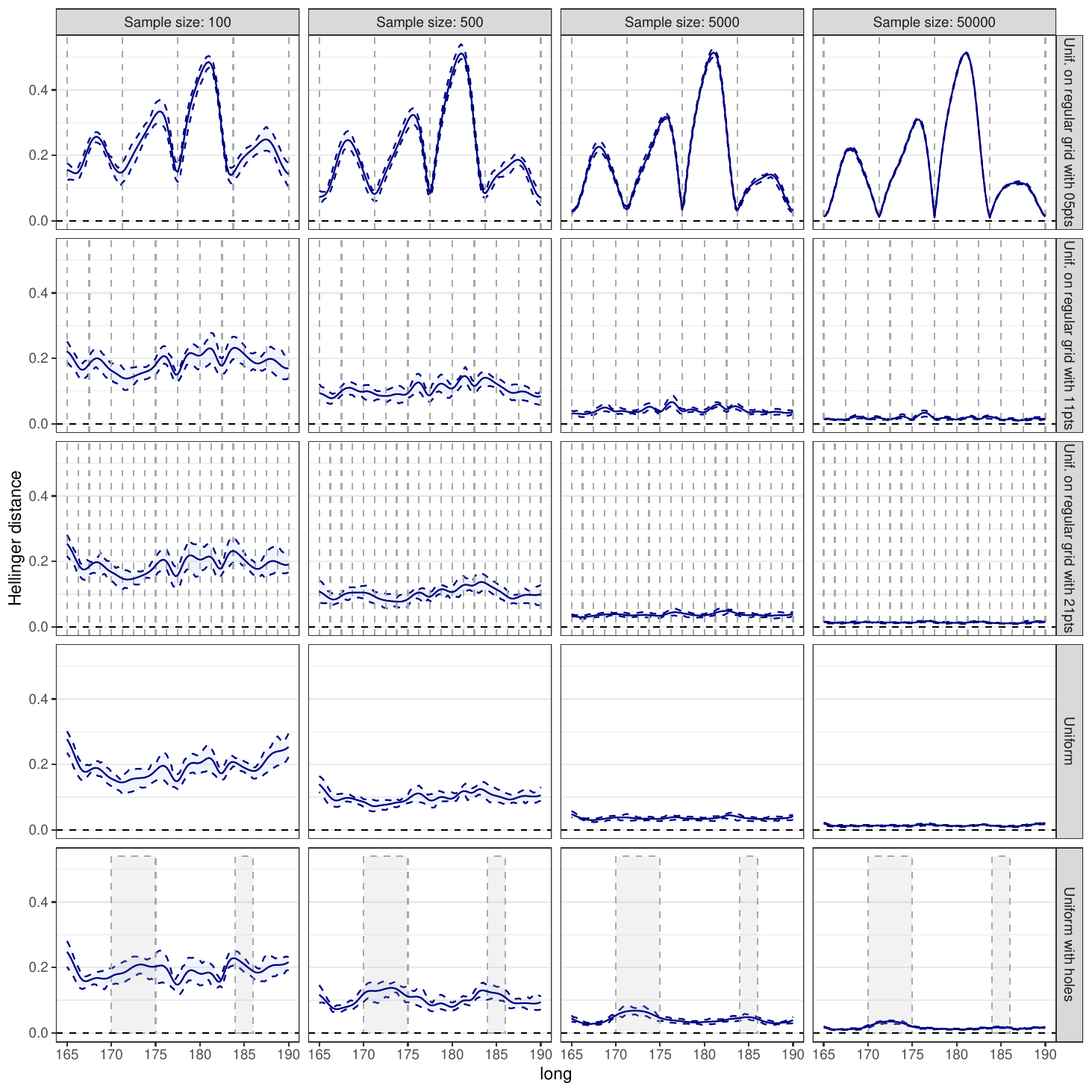}
\caption{Impact of the \longi selection scheme on SLGP density estimation using MAP. Local Hellinger distances between the estimated densities and the reference field, across varying training sample sizes: mean (solid lines) and 25th-75th percentile range (shaded ribbon).}
\label{fig:benchmark3b}
\end{figure}

\begin{figure}[H]
\centering
    \includegraphics[width=0.92\linewidth]{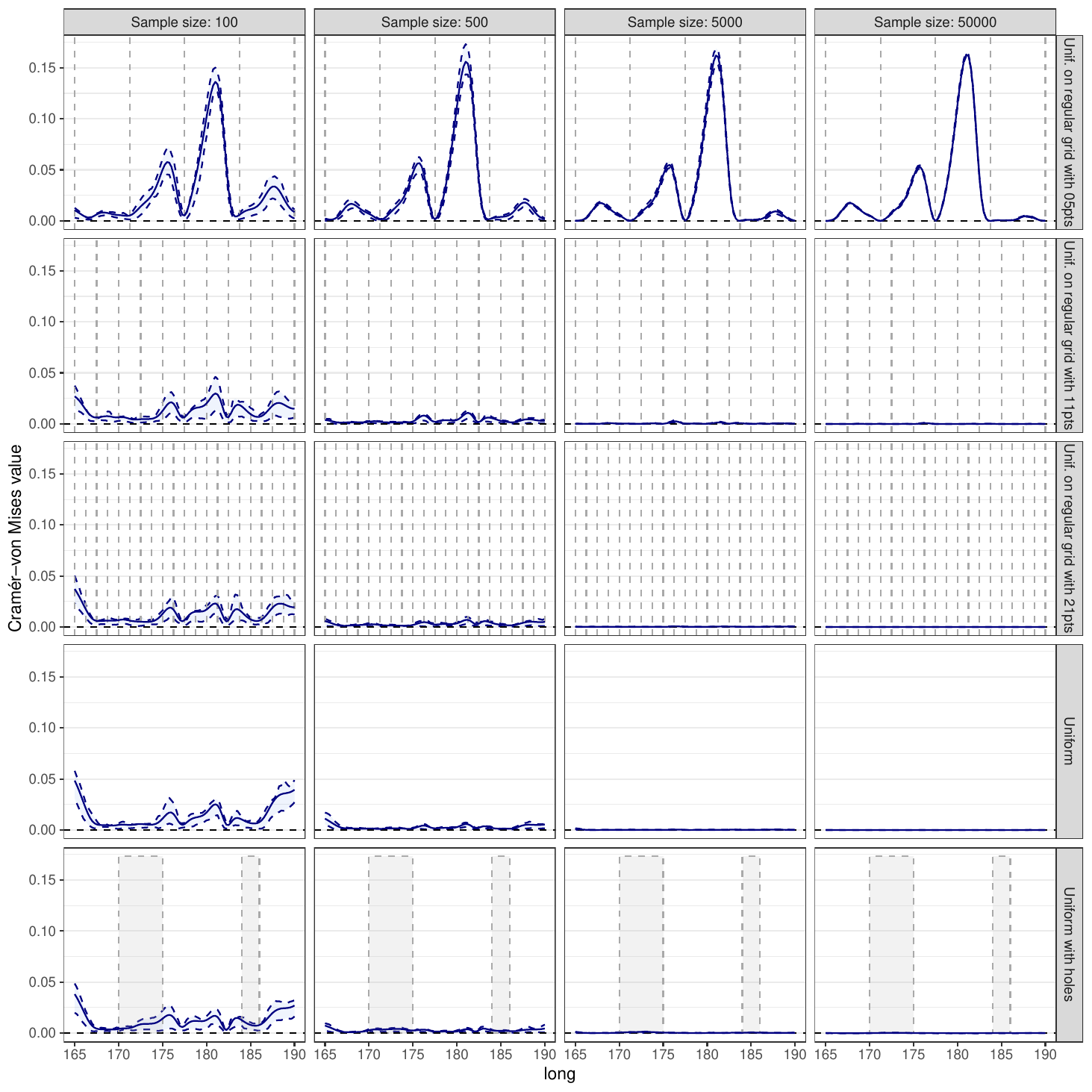}
\caption{Impact of the \longi selection scheme on SLGP density estimation using MAP. Local Cramér-von Mises criteria between the estimated densities and the reference field, across varying training sample sizes: mean (solid lines) and 25th-75th percentile range (shaded ribbon).}
\label{fig:benchmark3c}
\end{figure}

\begin{figure}[H]
\centering
    \includegraphics[width=0.92\linewidth]{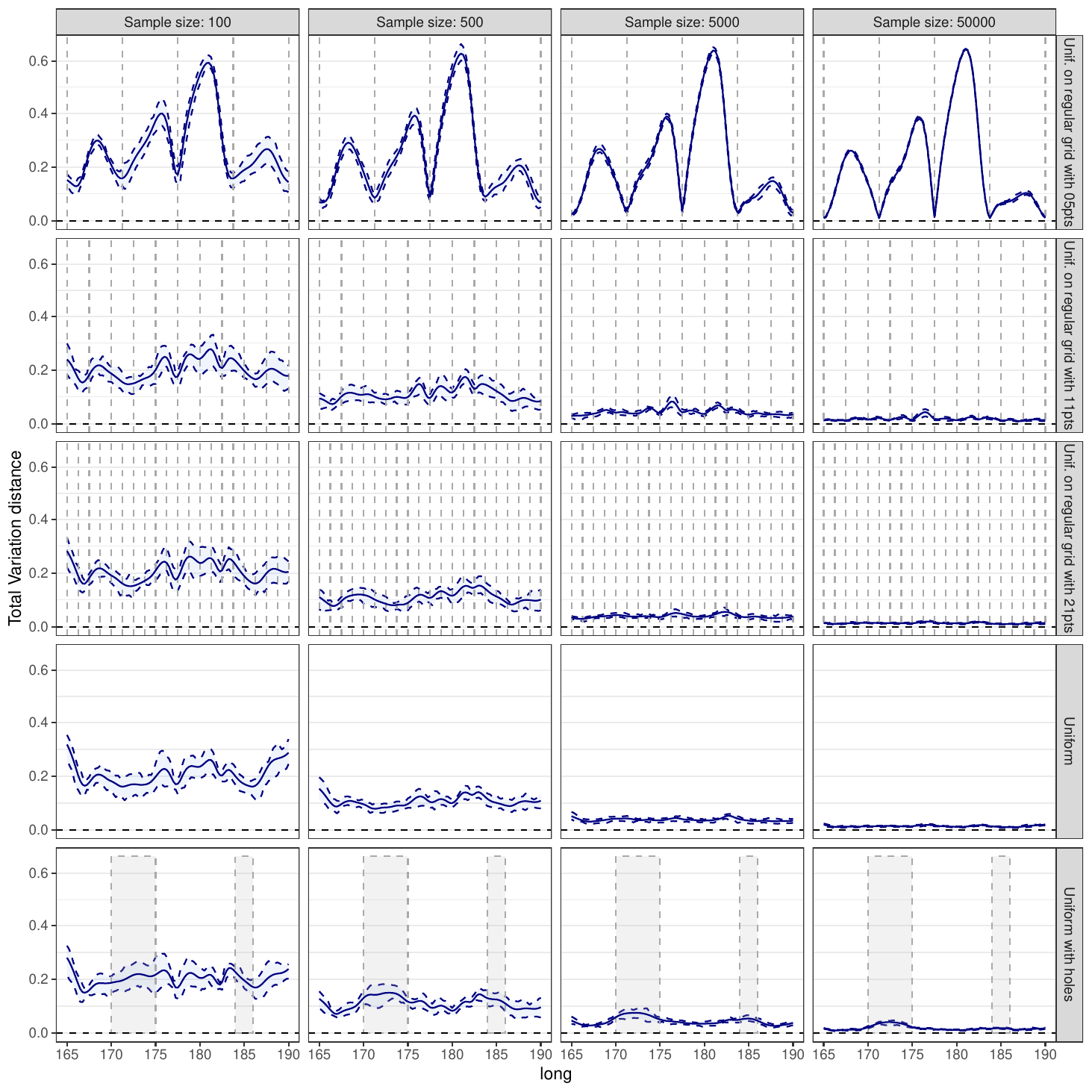}
\caption{Impact of the \longi selection scheme on SLGP density estimation using MAP. Local total variation distances between the estimated densities and the reference field, across varying training sample sizes: mean (solid lines) and 25th-75th percentile range (shaded ribbon).}
\label{fig:benchmark3d}
\end{figure}

\newpage
\section{Further views of the computational cost}
\label{app:timing}
Figure~\ref{fig:timingMethod} compares the three estimation methods at a fixed
integral scheme, and Figure~\ref{fig:timingRatio} the share of the total time
each of them spends in the estimation phase. A share of one means that the
pre-computation is negligible, which is the case for MCMC throughout, one half
means that setup and estimation weigh the same. For MAP the share lies between
0.08 and 0.44, so the pre-computation (and hence the choice of integral
scheme) is the dominant cost of a point estimate.

\begin{figure}[H]
\centering
    \includegraphics[width=0.95\linewidth]{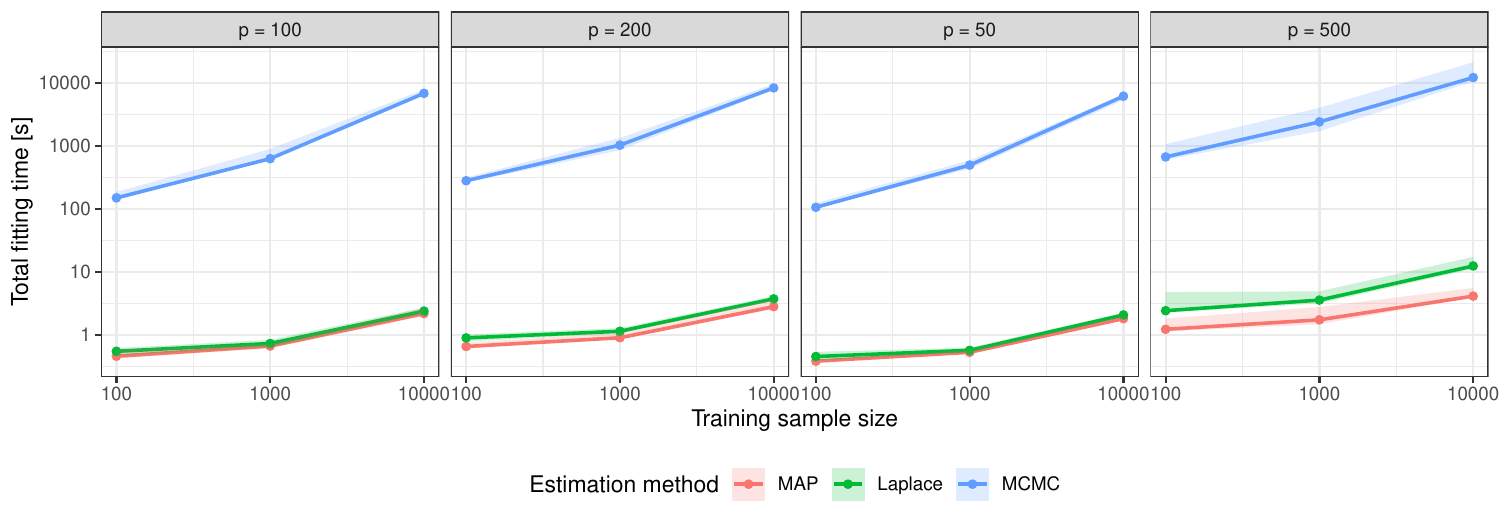}
\caption{Total fitting time of MAP, Laplace and MCMC against the sample size, at four ranks, for the \code{WNN} scheme. Medians and 10th-90th percentiles over ten replicates, logarithmic axes.}
\label{fig:timingMethod}
\end{figure}

\begin{figure}[H]
\centering
    \includegraphics[width=0.95\linewidth]{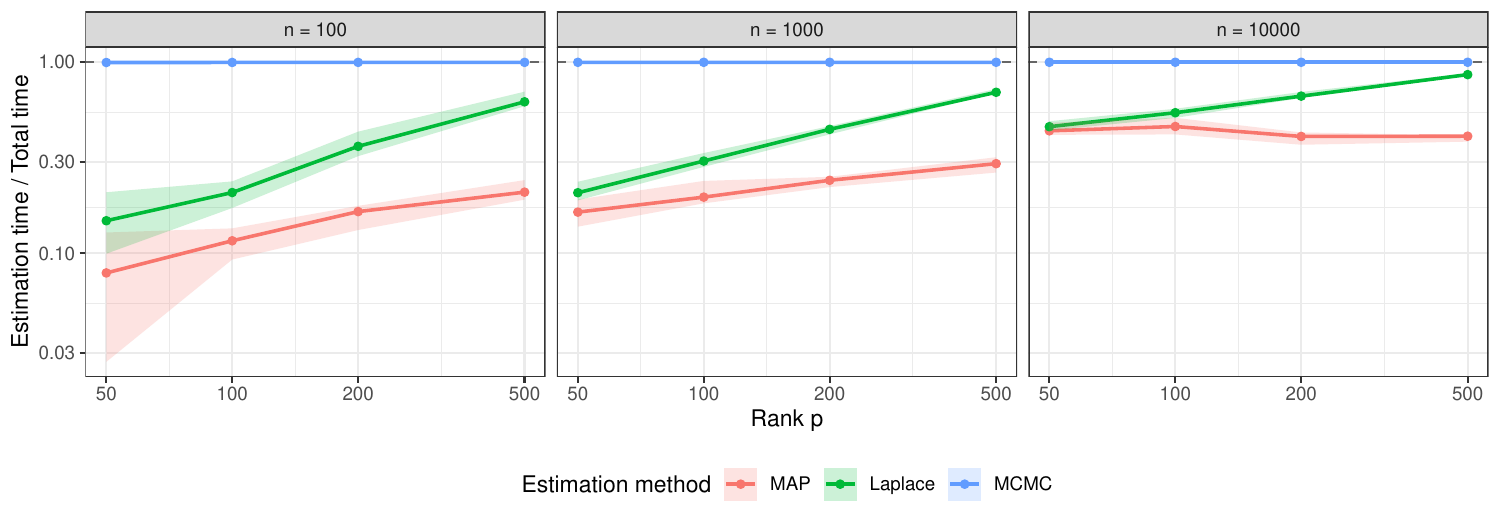}
\caption{Share of the total fitting time spent in the estimation phase, against
the rank, for the three methods.}
\label{fig:timingRatio}
\end{figure}
On the log-log scale the estimation time grows with slope 0.3 to 1.0 in $p$ for
MAP and MCMC, and 1.1 to 1.5 for Laplace, which forms the $p \times p$ Hessian
once, none approaches 2, so no operation of cubic cost dominates over this range of ranks.

\end{document}